\documentclass[11pt]{article}
\usepackage {graphicx}
\usepackage{amssymb}
\usepackage{amsmath}
\usepackage{amsbsy}
\usepackage{setspace}
\usepackage{natbib}
\usepackage{bm}
\usepackage{url}
\usepackage{wasysym}
\usepackage{textcomp}
\usepackage{mathtools}
\usepackage{verbatim}
\usepackage{amsthm}

\usepackage{mathtools}

\DeclarePairedDelimiter\floor{\lfloor}{\rfloor}

\usepackage{verbatim}
\usepackage{mathtools}
\usepackage{caption}
\usepackage{subcaption}

\newtheorem{theorem}{Theorem}
\newtheorem{proposition}{Proposition}
\newtheorem{lemma}{Lemma}
\newtheorem{corollary}{Corollary}
\newtheorem{remark}{Remark}
\newtheorem{example}{Example}

\write18{}

\newcommand{\beq}{\begin{equation}}
\newcommand{\eeq}{\end{equation}}

\newcommand{\goto}{\rightarrow}

\newcommand\X{\bm{X}}

\newcounter{rcnt}[section]

\newcommand{\beas}{\begin{eqnarray*}}
\newcommand{\eeas}{\end{eqnarray*}}
\newcommand{\bea}{\begin{eqnarray}}
\newcommand{\eea}{\end{eqnarray}}
\newcommand{\bei}{\begin{itemize}}
\newcommand{\eei}{\end{itemize}}
\newcommand{\ben}{\begin{enumerate}}
\newcommand{\een}{\end{enumerate}}
\newcommand{\bet}{\begin{theorem}}
\newcommand{\eet}{\end{theorem}}
\newcommand{\bel}{\begin{lemma}}
\newcommand{\eel}{\end{lemma}}
\newcommand{\bep}{\begin{proposition}}
\newcommand{\eep}{\end{proposition}}
\newcommand{\bed}{\begin{definition}}
\newcommand{\eed}{\end{definition}}
\newcommand{\bec}{\begin{corollary}}
\newcommand{\eec}{\end{corollary}}
\newcommand{\bex}{\begin{example}}
\newcommand{\eex}{\end{example}}

\newcommand{\ep}{\epsilon}

\def\0{\boldsymbol{0}}

\def\X{\boldsymbol{X}}

\newbox\TempBox \newbox\TempBoxA

\def\pr{\mathbb{P}} 
\def\ep{\mathbb{E}} 
\def\Var{\textsf{Var}} 

\def\0{\boldsymbol{0}}

\def\X{\boldsymbol{X}}

\newcommand{\khx}{K_{h_x}}
\newcommand{\kht}{K_{h_t}}

\begin{document}

\title{Structure--Adaptive Sequential Testing for Online False Discovery Rate Control}

\author{Bowen Gang$^1$, Wenguang Sun$^2$, and Weinan Wang$^3$}
\date{}

\maketitle

\begin{abstract}
 
		Consider the online testing of a stream of hypotheses where a real--time decision must be made before the next data point arrives. The error rate is required to be controlled at {all} decision points. Conventional \emph{simultaneous testing rules} are no longer applicable due to the more stringent error constraints and absence of future data. Moreover, the online decision--making process may come to a halt when the total error budget, or alpha--wealth, is exhausted. This work develops a new class of structure--adaptive sequential testing (SAST) rules for online false discover rate (FDR) control. A key element in our proposal is a new alpha--investment algorithm that precisely characterizes the gains and losses in sequential decision making. SAST captures time varying structures of the data stream, learns the optimal threshold adaptively in an ongoing manner and optimizes the alpha-wealth allocation across different time periods. We present theory and numerical results to show that the proposed method is valid for online FDR control and achieves substantial power gain over existing online testing rules.     

\end{abstract}

\noindent \textbf{Keywords:\/} Alpha--investing; Conditional local false discovery rate; Covariate--assisted inference;  Structured multiple testing; Time series anomaly detection

\footnotetext[1]{Department of Mathematics, University of Southern California.}  

\footnotetext[2]{Department of Data Sciences and Operations, University of Southern California. Corresponding Email: wenguans@marshall.usc.edu. The research of Wenguang Sun was supported in part by NSF grant DMS-1712983. } 

\footnotetext[3]{Snap Inc. }

\newpage

\section{Introduction}

The online testing problem is concerned with the investigation of a possibly infinite stream of null hypotheses $\{H_1, H_2, \cdots\}$ in an ongoing manner based on sequentially collected data $\{X_1, X_2, \cdots \}$. At each time point, the investigator must make a real-time decision after $X_t$ arrives, without knowing future data $\{X_{t+1}, X_{t+2}, \cdots\}$. The control of multiplicity in sequential testing typically involves imposing serial constraints on error rates over time, which requires that, for example, the family wise error rate (FWER) or false discovery rate (FDR; \citealp{BenHoc95}) must fall below a pre--specified level $\alpha$ at \emph{all} decision points.

The online testing problem may arise from a range of applications. For example, the quality preserving database (QPD) framework \citep{Aha10} has been widely employed by many research teams from diverse backgrounds. Some notable databases include Stanford's HIVdb that serves the community of anti-HIV treatment groups, WTCCC's large-scale database that is distributed to assist various whole-genome association studies, and the National Health Institute (NIH) influenza virus resource (IVR) that has been intensively queried by numerous researchers for designing new vaccines and treatments. The proper and efficient management of these large databases calls for new analytical tools for handling thousands of hypothesis tests with real--time decisions made in a sequential fashion. For instance, the NIH IVR has been used to investigate thousands of biomedical hypotheses and, per the record in PubMed, has lead to more than 1,000 scientific publications as of January 2020. 
It has become increasingly important to develop a powerful and effective monitoring system to  control the false positive findings over time.  
Another important application scenario, which is frequently encountered in finance, social media and mobile computing, is the real--time detection of anomalies based on high--frequency and large--scale time series data. For example, large travel service providers closely monitor the number of changes or cancellation requests of existing itineraries. An abnormal spike usually signifies an unexpected event. It is important for the company to detect such events early and make necessary adjustments. The development of online detection system plays a key role for providing novel and timely marketing insights and avoiding adverse financial losses.

Large-scale testing under the online setup poses several new issues that are not present in conventional ``offline'' setup. First, a real--time decision must be made before the next data point arrives. This makes conventional step--wise testing methods no longer applicable. For instance, the well--known Holm's procedure \citep{holm1979simple} for FWER control and Benjamini--Hochberg's procedure for FDR control both involve first ordering \emph{all} observed $p$-values and then choosing a threshold along the ranking. However, the ranking step becomes impossible due to the absence of future data. Second, in contrast with conventional FWER and FDR criteria that only require an overall assessment of the multiplicity in simultaneous testing, the online methods must proceed with more stringent error constraints that are  imposed sequentially at every decision point. This not only  leads to decreased power in detecting signals but also calls for more carefully designed online testing rules. Third, the data stream often encodes useful local structures, including signal magnitudes, sparsity levels and grouping patterns, that may vary over time. It is crucial to develop flexible and adaptive online rules to exploit the underlying domain knowledge and informative structures. Fourth, the online decision-making process, which proceeds sequentially without the knowledge of future, may come to a halt when the total error budget, or alpha--wealth, is exhausted. As a result, the investigator may miss all potential discoveries in the future. This concern must be carefully addressed because in many applications the hypothesis tests are conducted in an ongoing manner with unpredictable patterns -- even the total number of hypotheses to be investigated can be unknown. Finally, how to wisely allocate and invest the alpha--wealth to ensure the validity in error control while maintaining high statistical power of online testing rules in the long run has remained as a key issue that requires much research.

The online FDR control problem has received much recent attention and great progresses have been made. The alpha-investing (AI) idea \citep{Fos08} and its various generalizations \citep{AhaRos14, Rametal17,Jav16} have served as the basic framework and proved to be effective. Carefully designed AI rules are capable of handling an infinite stream of hypotheses and incorporating informative domain knowledge into the dynamic decision-making process. Beginning with a pre--specified alpha--wealth, the key idea in AI algorithms is that each rejection gains extra alpha--wealth, which may be subsequently used to make more discoveries at later time points. The generalized AI (GAI) algorithms \citep{AhaRos14,Rob18,Lyn17} are developed for a wider class of pay-out functions, enabling the construction of new online rules with increased power. The GAI++ framework \citep{Rametal17} improves the power of GAI methods uniformly and is capable of dealing with more general settings. The new class of weighted GAI++ methods are flexibly designed to allow ``indecisions'' and are capable of integrating prior domain knowledge. To alleviate the ``piggybacking'' and ``alpha--death'' issues of AI rules, \cite{Rametal17} discussed the concept of decaying memory FDR.  To effectively incorporate structural information into online inference, the SAFFRON procedure \citep{Rametal18} derived a sequence of thresholds that are adaptive to estimated sparsity levels and showed that the power can be much improved.

This article develops a new class of structure--adaptive sequential testing (SAST) rules for online FDR control with several new features. First, in contrast with existing AI and GAI rules whose building blocks are $p$-values, the class of SAST rules are built upon the conditional local false discovery rate (Clfdr), which optimally adapts to important local structures in the data stream. Second, the sequential rejection rule based on Clfdr leads to a novel alpha--investing framework that is fundamentally different from that in \cite{Fos08}. The new framework precisely characterizes the tradeoffs between different actions in online decision making, which provides key insights for designing more powerful online FDR rules. The new AI framework also reveals that SAST automatically avoids the ``alpha--death'' issue in the sense that its operation always reserves budget to reject new hypotheses, and can proceed in an ongoing manner to any time point in the future. Finally, by adaptively learning from past experiences and dynamically allocating the alpha–wealth, SAST can effectively avoid the ``piggybacking'' issue and improve its performance as more data are acquired. Our theoretical and numerical results demonstrate that SAST is effective for online FDR control, and achieves substantial power gain over existing methods in many settings. 

The article is organized as follows. Section 2 first introduces the model and problem formulation, and then develops the oracle SAST procedure for online FDR control by assuming that model parameters are known. Section 3 discusses computational algorithms, proposes the data-driven SAST rule and establishes its theoretical properties. Simulation is conducted in Section 4 to investigate the finite sample performance of SAST and compare it with existing methods.  SAST is illustrated in Section 5 through applications for identifying differentially expressed genes and detecting anomalies in time series data. The proofs are provided in the online supplementary material.

\setcounter{equation}{0}

\section{Oracle and Adaptive Rules for Online FDR Control}

We first describe the model and problem formulation in Section \ref{model.subsec}, then discuss three key elements in the proposed SAST rule in turn: a new test statistic to capture the structural information in the data stream (Sections 2.2 and 2.3); a new alpha--investing framework to characterize the gains and losses in sequential decision making (Section 2.4); and a new adaptive learning algorithm to optimize the alpha--wealth allocation (Sections 2.5).

\subsection{Model and Problem Formulation}\label{model.subsec}

Denote $\mathcal T$ a continuous temporal domain and $t\in\mathcal T$ a time point. Let $\mathbb T\subset \mathcal T$ be a discrete, ordered and evenly spaced index set for time labels\footnote{$\mathbb T$ may be taken either as  $\{1, 2, \cdots, t\}$ on a growing domain or a set of points that lie on a fixed-domain regular grid: $\{\frac{1}{t},\frac{2}{t},...,\frac{t-1}{t}, 1\}$ with $t\rightarrow \infty$. }.  Suppose we are interested in testing a sequence of null hypotheses $\{H_t: t\in \mathbb T\}$ based on data stream $\pmb X=(X_t: t\in \mathbb T)$. To describe the true states of nature, define Bernoulli variables $\theta_t$, where $\theta_t=0/1$ if $H_t$ is true/false. Let $\{\pi_t\equiv P(\theta_t=1): t\in\mathcal T\}$ denote the local sparsity levels that may vary over time. The observations can be described using a hierarchical model:
\begin{equation}\label{hmix-model}
\theta_t  \sim  \text{Bernoulli}(\pi_t), \quad X_t|\theta_t  \sim F_t=(1-\theta_t)F_0+\theta_t F_{1t},
\end{equation}
where $F_0$ and $F_{1t}$ are the null and non-null distributions, respectively. Denote $f_0$ and $f_{1t}$ the corresponding density functions. We assume that $F_0$ is known and identical for all $t\in \mathcal{T}$. By contrast, $\pi_t$ and $f_{1t}$ can vary smoothly in $t\in\mathcal T$. 
\begin{remark}\rm{
The inhomogeneity assumption reflects that signals may either vary in strengths or arrive at different rates over time. This structural information can be highly informative. The smoothness assumption makes it possible for pooling information from the observations in the neighborhood of $t$. We do not impose further assumptions on $\pi_t$ and $F_{1t}$, both of which will be estimated non-parametrically.}  
\end{remark}

Let $\pmb X^t=(X_i: i\in \mathbb T; \; i\leq t)$ be the collection of summary statistics (e.g. $p$--values or $z$--values) up to time $t$. Consider a class of online decision rules $\pmb{\delta}=\{\delta_t(\pmb X^t): t\in \mathbb T\}\in \{0, 1\}^{\mathbb T}$, where $\delta_t(\pmb X^t)$ represents a \emph{real-time decision} in the sense that $\delta_t$ only depends on information available at time $t$, with $\delta_t=1$ indicating that $H_t$ is rejected and $\delta_t=0$ otherwise. Denote $\pmb\delta^t=\{\delta_i(\pmb X^i): i\in \mathbb T; \; i\leq t\}$ the collection of decisions up to $t$. The online FDR problem is concerned with the performance of a stream of real--time decisions. For decisions up to $t$, let
\begin{equation}\label{FDR-t}
\text{FDR}^{t}(\pmb \delta^t)=\mathbb{E}\left\{\frac{\sum_{i\leq t; i\in \mathbb T}(1-\theta_i)\delta_i}{(\sum_{i\leq t; i\in \mathbb T}\delta_i)\vee 1}\right\},
\end{equation}
where the superscript ``t'' denotes that the FDR is evaluated at a specific time point. The goal is to construct a real--time decision rule $\pmb{\delta}=\{\delta_t(\pmb X^t): t\in \mathbb T\}$ that controls the $\text{FDR}^{t}$ at level $\alpha$ for all $t \in \mathbb T$. 
To compare the power of different testing rules, define the average power (AP) and missed discovery rate (MDR) as
\begin{equation}\label{AP:MDR}
\mbox{AP}^t(\pmb\delta^t)=\frac{\mathbb{E}(\sum_{i\leq t; i\in \mathbb T} \theta_i \delta_i)}{\mathbb{E}(\sum_{i\leq t; i\in \mathbb T}\theta_i)}; \quad  \mbox{MDR}^t(\pmb\delta^t)=1-\mbox{AP}^t(\pmb\delta^t).
\end{equation}

To simplify the discussion, throughout this section we assume that the distributional information such as the non-null proportion $\pi_{t}$ and density function $f_{t}$ in Model \ref{hmix-model} are known. Section 3 considers the case where model parameters are unknown and discusses in detail related estimation and implementation issues. 

\subsection{The oracle rule for simultaneous testing}\label{oracle-rule.subsec}

The goal of this section is to justify the fundamental role of Clfdr as the building block of the proposed online FDR rule.  

The online decision-making process is complicated due to the serial constraints on FDR and absence of future data. To focus on the essential issue, we first consider an ideal setup where a hypothetical oracle observes \emph{all data in a local neighborhood at once} and makes \emph{a batch of simultaneous decisions}. Let $d$ denote the size of a neighborhood. Consider the collection of hypotheses in a neighborhood prior to $t^*\geq d$: $\{H_i: t^*-d+1\leq i\leq t^*\}$. Denote the neighborhood $\mathcal N_d(t^*)= \{t^*-d+1, \cdots, t^*\}$ and the simultaneous decisions $\pmb\delta^*=\left\{\delta_i^*: i\in \mathcal N_d(t^*)\right\}$, where $\delta_i^*$ is allowed to depend on the entire $d$-vector $\pmb X^*=\{X_i: i\in \mathcal N_d(t^*)\}$. Unlike \eqref{FDR-t}, we only require that the FDR is controlled for the $d$ simultaneous decisions:
\beq\label{FDR-s}
\mbox{FDR}^{\rm s}(\pmb\delta^*)=\mathbb{E}\left\{\frac{\sum_{i\in \mathcal N_d(t^*)}(1-\theta_i)\delta_i}{(\sum_{i\in \mathcal N_d(t^*)}\delta_i)\vee 1}\right\},
\eeq
where the superscript ``s'' indicates a simultaneous--type FDR concept.

The simultaneous testing of multiple hypotheses can be conceptualized as a two-stage inferential process: firstly ranking all hypotheses according to a significance index and secondly choosing a cutoff along the ordered sequence. This process can be described by a thresholding rule of the form 
$$\pmb\delta=\{\mathbb I(\Lambda_i\leq c): i\in \mathcal N_d(t^*)\}, 
$$
where $\mathbb I(\cdot)$ is an indicator function, $\Lambda_i$ is the significance index of $H_i$ and $c$ is the cutoff of $\Lambda_i$. For example, the BH procedure uses the $p$-value as the significance index to order the hypotheses, and implements a step-up algorithm to determine a data-driven cutoff $c$. 

However, the $p$-value is inefficient for online FDR analysis as it fails to capture the important structural information in the data stream. We propose to use the conditional local false discovery rate (Clfdr)  as the significance index to order the hypotheses:
\beq\label{Clfdr}
\text{Clfdr}_t(x_t)=\mathbb{P}(\theta_t=0|X_t=x_t)=\frac{(1-\pi_{t})f_0(x_t)}{f_{t}(x_t)}, \;\mbox{for $t\in\mathbb T$}.
\eeq

Denote $\text{Clfdr}_{(1)},\cdots,\text{Clfdr}_{(d)}$ the ordered Clfdr values in $\mathcal N_d(t^*)$ and $H_{(1)}, \cdots, H_{(d)}$ the corresponding hypotheses. To determine the cutoff for simultaneous testing, we apply a step-wise algorithm 
\beq\label{AZ}
k=\max\left\{j: \frac 1 j\sum_{i=1}^j\text{Clfdr}_{(i)}\leq \alpha\right\}.
\eeq
Then the threshold is $c=\text{Clfdr}_{(k)}$ and we reject $H_{(1)}, \cdots, H_{(k)}$. The Clfdr rule \eqref{AZ} may be viewed as an oracle rule that sees all data in a local neighborhood at once and then makes simultaneous decisions. In Appendix \ref{opt-clfdr.sec}, we establish the optimality property of the Clfdr rule for simultaneous testing under the ``offline'' setup. An infinite data stream can be approximately by sequential data points arrived in batches. Intuitively, the Clfdr statistic provides a good building block for developing new online sequential testing rules as it is optimal for simultaneous inference in each batch of data points. 
 
\begin{remark}\rm{In the ``offline'' setup for simultaneous testing with a covariate sequence, which includes the Clfdr rule \eqref{AZ} as a special case, \cite{Caietal19} develops asymptotic optimality theory. We can similarly show that  \eqref{AZ} is asymptotically optimal in the sense that it achieves the benchmark of a hypothetical oracle. However, the optimality issue in the online setup, which depends on many other factors such as the optimal allocation of alpha--wealth and prediction of future patterns over time, is still an open issue and requires much research.
}
\end{remark}

\subsection{Adapting to local structures by Clfdr: an illustration}\label{adapt:sec}

The incorporation of structural information and domain knowledge promises to improve the power of existing FDR procedures (\citealp{Genetal06, CaiSun09, Huetal10, LeiFit18, Caietal19}). For example, the works by \cite{Huetal10}, \cite{LiBar19} and \cite{Xiaetal19} showed that the weighted $p$-values can be constructed to capture the varying sparsity levels of ordered or grouped hypotheses. In contrast with the $p$-value, the Clfdr takes into account important structural information such as $\pi_t$ and $f_t$, which makes Clfdr an ideal building block for multiple testing with inhomogeneous data streams. We present an example to illustrate the advantage of the Clfdr rule. 

 Consider the following situation where the data stream $\{X_1, X_2, \ldots, X_t, \ldots\}$ obeys a random mixture model with varying sparsity levels:
\begin{equation}\label{hmix-model2}
X_t \sim (1-\pi_{t})N(0,1)+\pi_{t} N(\mu, 1).
\end{equation}
Model \eqref{hmix-model2} is a special case of Model \eqref{hmix-model}: the null and alternative densities are fixed and the dynamic part is fully captured by the varying proportion $\pi_t$. The key idea of Clfdr and weighted p-value (in the form of $p_t/w_t$, where $w_t$ is the weight for $H_t$) is to up–weight the hypotheses in a local neighborhood where signals appear more frequently (e.g. in clusters). 

To compare the effectiveness of different weighting methods, we simulate a data stream for testing $m=5000$ hypotheses. The top row in Figure \ref{clfdr:fig} sets $\pi_t=0.5$ in blocks $[1001:1150]$, $[2001:2150]$, $[3001:3100]$ and $[4001:4150]$, and $\pi_t=0.01$ elsewhere. We vary $\mu$ from 2 to 4. The bottom row sets $\mu=2.5$ and vary $\pi_t$ from 0.2 to 0.9 in the above blocks. The block structure is highly informative and can be exploited by Clfdr and weighted p-values to improve the power. We apply the following methods at FDR level $\alpha=0.05$ by assuming that the model parameters in \eqref{hmix-model2} are known: BH \citep{BenHoc95}, the structure--adaptive BH algorithm (SABHA; \citealp{LiBar19}) using weighted $p$-values with $w_t=1/(1-\pi_t)$, the GAP method \citep{Xiaetal19} using weighted $p$-values with $w_t=\pi_t/(1-\pi_t)$, and the Clfdr rule \eqref{AZ}. We can see that all methods control the FDR at the nominal level. In terms of the power, BH can be improved by SABHA and GAP, both of which are dominated by the Clfdr rule. Clfdr captures the varying structure in the data stream more effectively: in addition to varied $\pi_t$, it also adapts to $f_t$, leading to further power improvement. 

\begin{figure}
\includegraphics[width=0.95\textwidth]{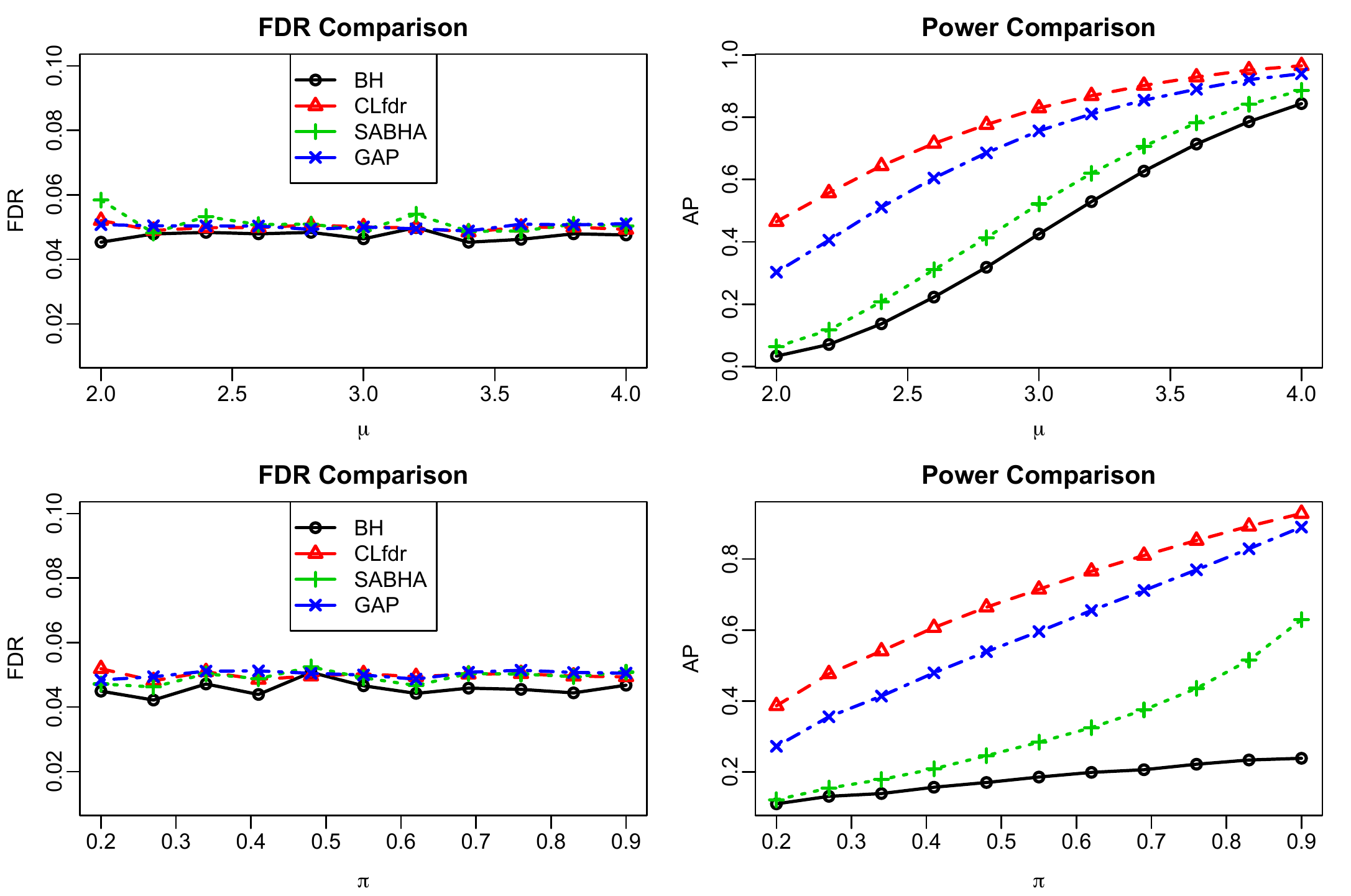}\caption{\footnotesize Structure--adaptiveness: Clfdr vs weighted p-values.}\label{clfdr:fig}
\end{figure}

\subsection{A new alpha--investing framework}\label{alpha-invest.sec}

Existing FDR methods such as the BH and Clfdr procedures are simultaneous inference procedures that involve first ordering the significance indices ($p$-value or Clfdr) of all hypotheses and then applying a step-wise algorithm to the ordered sequence to determine the threshold. However, the ranking and thresholding strategy cannot be applied to the online setting where the investigator must make real--time decisions without seeing future observations. This section discusses how to avoid the overflow of the FDR at any given time $t$ and how to efficiently allocate the alpha--wealth to increase the power.

We start with a novel interpretation of the alpha--investing idea by recasting the Clfdr algorithm \eqref{AZ} as a varying--capacity knapsack process. Denote $\mathcal R_{t}\subset \{H_1, H_2, \cdots, H_t\}$ the collection of rejected hypotheses at time $t$. The decision process \eqref{AZ} can be conceptualized as a sequence of comparisons of two quantities: the nominal FDR level $\alpha$ and the average of the rejected Clfdr values. Specifically, \eqref{AZ} motivates us to consider the constraint
\begin{equation}\label{FDR-serial}
\texttt{Ave}\left\{\mbox{Clfdr}_i: i\in \mathcal R_t \right\}\leq \alpha, \; \mbox{ for all $t\in \mathbb T$,}
\end{equation} 
where $\texttt{Ave}(A)$ denotes the average of the elements in set $A$. The simultaneous testing setup is only concerned with one constraint at the last time point when all data have been observed. By contrast, the online setup poses a series of constraints, e.g. \eqref{FDR-serial} must be fulfilled for every $t$ to avoid the overflow of $\mbox{FDR}^t$ \eqref{FDR-t}.

We view \eqref{FDR-serial} as a dynamic decision process resembling a knapsack problem, where $H_t$ can only be rejected when the following constraint is satisfied: 
\begin{equation}\label{FDR-cons} 
 \mbox{Clfdr}_t - \alpha \leq C_t \coloneqq - \sum_{H_i\in \mathcal R_{t-1}} \left(\mbox{Clfdr}_i - \alpha\right), \mbox{ for $t=1, 2, \cdots$}
\end{equation}
where $C_t$ is the \emph{capacity} (of the knapsack) at time $t$ with the default choice $C_1=0$. The capacity may either expand or shrink over time, depending on the sequential decisions along the data stream. This dynamic process can be described as follows. The initial capacity is $C_1=0$. Starting from $t=1$, we reject $H_t$ if \eqref{FDR-cons} is fulfilled. 
If $H_t$ with $\mbox{Clfdr}_t<\alpha$ is rejected, then the capacity $C_t$ increases by $\alpha-\mbox{Clfdr}_t$ (gain);  hence we earn bonus room. By contrast, if $H_t$ with $\mbox{Clfdr}_t>\alpha$ is rejected, then $C_t$ decreases by $\alpha-\mbox{Clfdr}_t$ (loss). 

The decision process \eqref{FDR-cons} provides a new alpha--investing framework that precisely characterizes the gains and losses in sequential testing. In contrast with the alpha--investing framework in \cite{Fos08}, which views each rejection as a gain of extra alpha--wealth, the new characterization  \eqref{FDR-cons} reveals that not all rejections are created equal: rejections with small Clfdr will lead to increased alpha--wealth whereas rejections with large Clfdr will lead to decreased alpha--wealth. This view provides key insights for designing more powerful online FDR rules. Moreover, the new AI framework reveals that utilizing Clfdr rules can automatically avoid the ``alpha--death'' issue. Specifically, the process \eqref{FDR-cons} can always reject new hypotheses with $\mbox{Clfdr}<\alpha$ regardless of the current budget, and can proceed in an ongoing manner to any time point in the future.

\subsection{Oracle--assisted adaptive learning and the SAST algorithm}

To efficiently allocate the alpha--wealth, we need to further refine the online algorithm \eqref{FDR-cons} to avoid making imprudent rejections that can potentially eat up all the budget. The specific issue is referred to as ``piggybacking'' \citep{Rametal17}, which, in a vivid way, describes the phenomenon that a string of bad decisions were made due to previously acquired budget. 

To see the necessity of taking careful actions, suppose that we have accumulated some bonus room over time before observing a very large $\mbox{Clfdr}_t$ satisfying \eqref{FDR-cons}. Although rejecting $H_t$ is an action that obeys the FDR constraint, the action can be unwise since it is possible that we can invest the extra ``cost'', $\mbox{Clfdr}_t-\alpha$, to make more discoveries at later time points. A practical strategy is to incorporate a ``barrier'' $\gamma_t$ and modify \eqref{FDR-cons} as
\begin{equation}\label{FDR-cons2} 
\mbox{Clfdr}_t< \gamma_t \; \mbox{and}\; \mbox{Clfdr}_t - \alpha \leq - \sum_{H_i\in \mathcal R_{t-1}} \left(\mbox{Clfdr}_i - \alpha\right).
\end{equation}
The barrier can effectively prevent ``piggybacking'' by filtering out large $\mbox{Clfdr}_t$ and hence saving budget for future. 

The choice of $\gamma_t$ depends on the pattern of future hypotheses. However, all online methods must proceed without seeing the future. To resolve the issue, consider the oracle Clfdr rule \eqref{AZ} that sees all data in a local neighborhood at once. If we assume that the hypothesis stream is ``locally stable'' in its patterns, then $\gamma_t$ may be informed by the oracle rule \eqref{AZ} \emph{simultaneously} conducted on a local neighborhood $\mathcal N_d(t)=\{t-d-1, \cdots, t\}$. The rationale is to use recent past data to get some ideas about the patterns of hypotheses to arrive in the near future. Concretely, we first order  $\{H_i: i\in\mathcal N_d(t)\}$ according to their Clfdr values, then run the ``offline'' algorithm \eqref{AZ} to set the barrier $\gamma_t=\mbox{Clfdr}_{(k+1)}$. The online algorithm, by acting as if it sees the future, can effectively filter out large Clfdr values and hence avoid inefficient investments. The operation of algorithm \eqref{AZ} also implies that the barrier $\gamma_t$ may be either raised or lowered according to the varied $\pi_t$ and $f_t$ in the dynamic model, which is desirable in practice for dealing with inhomogeneous data streams. In Section \ref{bar:sec}, we illustrate that the incorporating of the barrier can greatly reduce the MDR \eqref{AP:MDR}. 

Finally, we present the proposed structure--adaptive sequential testing (SAST) rule (oracle version with known parameters) in Algorithm 1. The SAST algorithm essentially utilizes the sequential constraints \eqref{FDR-cons2} with barriers set by the offline algorithm \eqref{AZ}. 

\captionsetup[table]{labelformat=empty}
\begin{table}[ht]
\centering
 \caption*{\label{algorithm.table} Algorithm 1. The oracle SAST rule. 
 }
\begin{tabular}{l}
  \hline
  \hline
\textbf{Intialization: }$\mathcal{A}_0=\emptyset,\;\gamma_0=\alpha$.   
\\
\medskip

\textbf{Updating the barrier:} Let $\mathcal N_d(t)= \{t-d+1, \cdots, t\}$. Sort $\{\text{Clfdr}_i: i\in \mathcal N_d(t)\}$ from \\ the smallest to largest and denote the ordered statistics as $\{\mbox{Clfdr}_{(1)}^t, \mbox{Clfdr}_{(2)}^t, \cdots\}$. \\ If $\text{Clfdr}_{(1)}^t>\alpha$, keep the same barrier $\gamma_t=\gamma_{t-1}$. Otherwise let $k=\max\{j: Q^t(j)\leq \alpha\}$, \\ where $Q^t(j)=\frac{1}{j}\sum_{i=1}^j\text{Clfdr}_{(i)}^t$, and update the barrier as $\gamma_t=\mbox{Clfdr}_{(k+1)}^t$. 

\\
\medskip

\textbf{Decision: } Let $\mathcal{R}_t=\{i\leq t: \delta_i=1\}$ and denote $|\mathcal R_t|$ its cardinality. If $\text{Clfdr}_t< \gamma_t$ \\ and $\{|\mathcal{R}_{t-1}|+1\}^{-1}\left(\sum_{i\in\mathcal{R}_{t-1}}\text{Clfdr}_i+\text{Clfdr}_t\right)\leq \alpha$, then $
\delta_t=1$. Otherwise $\delta_t=0$.\\
\hline
\hline
\end{tabular}
\end{table}

We can see that Algorithm 1 runs two parallel procedures: an online procedure for making real--time decisions and an ``offline'' procedure for determining the barrier. Thus the information of every data point has  been used twice: first $X_t$ is used for real--time decision--making at time $t$, then $X_t$ is stored as past data so that we can ``learn from experiences'' via the offline oracle. The following theorem shows that Algorithm 1 is valid for online FDR control. 
\begin{theorem}
Consider the online FDR procedure $\pmb\delta=(\delta_t: t\in \mathbb T)$, where $\delta_t$ is determined by Algorithm 1. Denote $\pmb\delta^t=(\delta_i: i\leq t; i\in \mathbb T)$. Assume that the Clfdr values are known. Then we have $\mbox{FDR}^{\rm t} (\pmb\delta^t)\leq \alpha, \; \mbox{for all $t\in \mathbb T$. }$
\end{theorem}

\setcounter{equation}{0}

\section{Data-Driven SAST and Its Theoretical Properties}\label{SAST-DD.sec}

We first develop estimation methodologies and computational algorithms to implement the SAST rule in Section \ref{estimate.subsec}, then establish the theoretical properties of the data-driven procedure in Section \ref{SAST-theory.sec}.  
	
	\subsection{Data-driven procedure and computational algorithms}\label{estimate.subsec} 
	
We assume that the null distribution of $z$-values $f_0$ is known, which is a standard practice in the literature\footnote{In situations where the empirical null is more appropriate \citep{Efr04}, $f_0$ can be first estimated using the method in \cite{JinCai07} and then treated as known.}. The key quantities remained to be estimated are $\pi_t$ and $f_t(x)$. In our motivating applications such as queries of QPDs and anomaly detection in high--frequency time series, the databases or servers have already collected large amounts of data at the beginning of the online FDR analysis. Let $\{X_{-K_0}, \cdots, X_{-1}, X_0\}$ denote the available data and suppose we start online testing at $t=1$ with a data stream $\{X_1, X_2, \ldots\}$\footnote{In situations where the online FDR analysis must start without prior data, we suggest applying existing methods such as LOND first and then switch to SAST as more data are acquired.}.  
	
	The conditional density $f_t$ can be estimated using standard (one--sided) bivariate kernel methods \citep{Sil86}:
	\begin{equation}\label{cond-dens}
	\hat f_t(x) = \frac{\sum_{j=t-d+1}^{t-1} K_{h_t}\left({j-t}\right) K_{h_x}\left({x_j-x}\right)}{ \sum_{j=t-d+1}^{t-1} K_{h_t}\left({j-t}\right)},
	\end{equation}
	where $d \leq K_0$ is the length of the moving window that includes a pre-specified number of observations, $K(t)$ is a kernel function, $h_t$ and $h_x$ are the bandwidths, with $K_h(t)=h^{-1}K(t/h)$. 
	\begin{remark}\rm{
			In analysis of large-scale high-frequency time series data such as the NYC taxi data (Section~\ref{app:sec}), we can pre-specify $d$, say, to be 1000 to speed up the computation. This virtually has no impact on the estimator $\hat f_t$ (compared to using all previous data). Otherwise we can always set $d=t$. Note that our estimator has followed the standard practice in density estimation, which does not include $X_t$ when estimating $f_t(x)$ at time $t$. }
	\end{remark}

	Next we propose a weighted screening approach to estimate the unknown proportion $\{\pi_t: t\in \mathbb T\}$. The key idea is to use a kernel, which weights observations by their distance to $t$, to pool information from nearby time points. Let $h_t$ be the bandwidth\footnote{We recommend using the same $h_t$ in both  \eqref{cond-dens} and \eqref{q-tau-hat} to stabilize the performance.} and $K$  a kernel function satisfying $\int K(t)dt=1$, $\int tK(t)dt=0$ and $\int t^2K(t)dt <\infty$. Consider a screening procedure $\mathcal T_t(\tau)=\{t-d+1\leq i \leq t-1: P_{i}>\tau\}$, where $\tau$ is a pre-specified threshold. We propose the following estimator based on \cite{Caietal19}:
	\begin{equation}\label{q-tau-hat} 
	\hat \pi^\tau_t=1-\frac{\sum_{i\in \mathcal T_t(\tau)} K_{h_t}\left(t-i\right) }{(1-\tau)\sum_{i=t-d+1}^{t-1} K_{h_t}\left(t-i\right)}.
	\end{equation}
	
	Now we provide some intuitions of the estimator \eqref{q-tau-hat}. First, at time $t$, define $v_h(t, i)={K_{h_t}(|t-i|)}/{K_{h_t}(0)}$. We can view $m_t=\sum_{i=t-d+1}^{t-1} v_h(t, i)$ as the ``total'' number of observations at time $t$. Suppose we are interested in counting how many null $p$-values are greater than $\tau$ among the $m_t$ ``observations'' at $t$. The empirical count is given by 
	$\sum_{i\in \mathcal T_\tau} v_h(t, i)$, whereas the expected count is given by  
	$\{\sum_{i=t-d+1}^{t-1} v_h(t, i)\}\{1-\pi_t\}(1-\tau).$ Equation \eqref{q-tau-hat} can be derived by first setting equal the expected and empirical counts and then solving for $\pi_t$.  In Section \ref{SAST-theory.sec} we show that $\hat \pi^\tau_t$ is a consistent estimator of
	\beq\label{pi-tau}
	\pi^\tau_t = 1 - (1-\tau)^{-1}{\pr(P_t>\tau)},
	\eeq
	which always underestimates $\pi_t$ and guarantees (conservative) FDR control (Propositions~\ref{prop1}). 
	
\begin{remark}\rm{
		There is a bias-variance tradeoff in the choice of $\tau$ for the proposed estimator $\hat\pi^\tau_t$. We shall see that when $\tau$ increases, the ``purity'' of the screening subset $\mathcal T(\tau)$ increases, which decreases the approximation bias of $\pi^\tau_t$ (desirable). At the same time, when $\tau$ increases, the sample size for estimating $\pi^\tau_t$ will decrease, thereby increasing the variance of the estimator $\hat \pi^\tau_t$ (undesirable). The common choice of $\tau$ is 0.5. In Section \ref{implement:sec}, we discuss a data--driven algorithm that chooses $\tau$ adaptively.}
	
\end{remark}

	Combining (\ref{cond-dens}) and (\ref{q-tau-hat}), we propose to estimate the Clfdr as
	\beq\label{Clfdr-hat}
	\widehat{\text{Clfdr}}_t=\min\left\{\frac{(1-\hat{\pi}^\tau_t)f_0(x_t)}{\hat{f}_t(x_t)},1\right\}, \quad t\in \mathbb T.
	\eeq
	Our proposed data-driven rule implements Algorithm 1 by substituting $\widehat{\text{Clfdr}}_t$ in place of $\text{Clfdr}_t$. The data-driven algorithm is summarized in Algorithm 2.

	\captionsetup[table]{labelformat=empty}
	
	\begin{table}[ht]
		\centering
		\caption*{\label{algorithm.table} Algorithm 2. The data-driven SAST.
		}
		\begin{tabular}{|p{0.9\textwidth}|}
			\hline
			\hline
			\textbf{Initialization: }$\mathcal{R}_0=\emptyset,\;\gamma_0=\alpha$.   
		
			\medskip
			\textbf{Estimation: } 
$\widehat{\text{Clfdr}}_t=\min\left\{\frac{(1-\hat{\pi}_t^\tau)f_0(x_t)}{\hat{f_t}(x_t)}, 1\right\}, $ where $\hat \pi^\tau_t$ and $\hat f_t$ are defined by \eqref{q-tau-hat} and \eqref{cond-dens}, respectively.  
			
			\medskip
			
			\textbf{Updating the barrier:} Let $\mathcal N_d(t)= \{t-d+1, \cdots, t\}$. Sort $\{\widehat{\text{Clfdr}}_i: i\in \mathcal N_d(t)\}$ from the smallest to largest and denote the ordered statistics as $\{\widehat{\mbox{Clfdr}}_{(1)}^t, \widehat{\mbox{Clfdr}}_{(2)}^t, \cdots\}$. If $\widehat{\text{Clfdr}}_{(1)}^t>\alpha$, keep the same barrier $\gamma_t=\gamma_{t-1}$. Otherwise let $k=\max\{j: Q^t(j)\leq \alpha\}$, where $Q^t(j)=\frac{1}{j}\sum_{i=1}^j\widehat{\text{Clfdr}}_{(i)}^t$, and update the barrier as $\gamma_t=\widehat{\mbox{Clfdr}}_{(k+1)}^t$. 
			
			\medskip
			
			\textbf{Decision: } Let $\mathcal{R}_t=\{i\leq t: \delta_i=1\}$ and denote $|\mathcal R_t|$ its cardinality. If $\widehat{\text{Clfdr}}_t< \gamma_t$ \\ and $\{|\mathcal{R}_{t-1}|+1\}^{-1}\left(\sum_{i\in\mathcal{R}_{t-1}}\widehat{\text{Clfdr}}_i+\widehat{\text{Clfdr}}_t\right)\leq \alpha$, then $
			\delta_t=1$. Otherwise $\delta_t=0$.\\
			\hline
			\hline
		\end{tabular}
	\end{table}

	\subsection{Theoretical properties of data--driven SAST}\label{SAST-theory.sec}
	
	This section aims to show that the data--driven SAST procedure is asymptotically valid for online FDR control. Our theoretical analysis is divided into three steps. The first step (Proposition \ref{prop1}) shows that a hypothetical rule, which substitutes  
	\beq\label{Clfdr-tau}
	\mbox{Clfdr}_t^{\tau}=\dfrac{(1-\pi_t^\tau)f_0(x_t)}{f_t(x_t)}
	\eeq
	in place of $\mbox{Clfdr}_t$ in Algorithm 1, is conservative for online FDR control. 
	
	\begin{proposition}\label{prop1}
		Consider $\pi^\tau_t$ defined by \eqref{pi-tau}, then we have $\pi^\tau_t\leq \pi_t$ and $\mbox{Clfdr}_t\leq \mbox{Clfdr}_t^{\tau}$. Hence the hypothetical rule using \eqref{Clfdr-tau} is valid (and conservative) for online FDR control.
	\end{proposition}
	
	The second step (Proposition \ref{SAWS_prop}) shows that $\widehat{\mbox{Clfdr}_t}$ is a consistent estimator of $\mbox{Clfdr}_t^{\tau}$. We prove the result by appealing to the infill--asymptotics framework (\citealp{Ste12}), which converts the set of time points $\{1, 2, \cdots, t\}$ on a growing domain to a set of points that lie on a fixed-domain regular grid: $\{\frac{1}{t},\frac{2}{t},...,\frac{t-1}{t}, 1\}$. The discussions in \cite{Ste12} indicate that the in-fill model is equivalent to the growing domain model under mild conditions: When $t\rightarrow \infty$, the asymptotic arguments, which respectively correspond to letting the grid become denser and denser in the fixed interval $(0,1]$ and letting the domain $\{1, 2, \cdots, t\}$ to grow to infinity, can be essentially established in the same manner. We state the fixed domain theory as it naturally connects to the familiar density estimation theory, where the notations and regularity conditions are standard and easy to understand. The growing domain version of the theory is briefly discussed in Appendix~\ref{grow}.

	We can similarly define the bivariate density estimator and the following conditional proportion estimator: 
	\beq\label{pi-tau2}
	\hat{\pi}^{\tau}_t=1-\frac{\sum_{i\in \mathcal T_t(\tau)} K_{h_t}\left(1-i/t\right) }{(1-\tau)\sum_{i=t-d+1}^{t-1} K_{h_t}\left(1-i/t\right)}.
	\eeq 
The two estimators \eqref{q-tau-hat} and \eqref{pi-tau2} are essentially identical (with rescaled bandwidths). 	
	
We state the following regularity conditions. Condition (A1) requires that $f_t(x)$ is smooth in $t$. Conditions (A2) to (A4) are standard in density estimation theory; see, for example, \citep{WanJon94}. 

\noindent\textbf{(A1)}: For any $s\in(0,1]$ and $\epsilon>0$, $\exists \delta$ such that if $|s-s'|\leq \delta, s'\in(0,1]$ then
		$\int|f_s(x)-f_{s'}(x)|dx<\epsilon$.
		
\noindent\textbf{(A2)}: $h_x\rightarrow 0$, $h_t\rightarrow 0$ and $th_xh_t\rightarrow \infty$.
				
\noindent\textbf{(A3)}: $f_j(x)<C$ and $\int |f^{''}_j(x)|dx<C $  for all $j$.

\noindent\textbf{(A4)}: $dh_t\rightarrow \infty$ and $d\geq cth_t$ for some $c>0$.
	
	\begin{proposition}\label{SAWS_prop}
		Suppose (A1)--(A4) hold, then $\widehat{\text{Clfdr}}_t\xrightarrow{p} \text{Clfdr}_t^\tau.$ 
	\end{proposition}

	In the third step of our theoretical analysis (Theorem \ref{thm2}), we establish the asymptotic validity of the data-driven SAST procedure for online FDR control.

\begin{theorem}\label{thm2}
			Assume the conditions in Proposition \ref{SAWS_prop} hold. Then for any given time $t$, the data-driven SAST rule (Algorithm 2) controls the $\text{FDR}^t$ at level $\alpha$ asymptotically.
		\end{theorem}

	\subsection{Theory for data streams with fixed distributions}
	
	SAST learns from past decisions and improves its performance over time through the assistance from an offline oracle. The barrier $\gamma_t$ would become more informative as more tests are conducted. Specifically, the initial barrier is set to be $\alpha$ at time $t=1$, which is very conservative. In the special case when the mixture model has fixed $\pi_t$ and $f_t$ over time, we can show that the barrier $\gamma_t$ would converge to $\gamma_{OR}$, where $\gamma_{OR}$ is the optimal threshold of the ``offline'' oracle procedure in Section \ref{oracle-rule.subsec}. Hence, provided that the capacity allows, the operation of \eqref{FDR-cons2} 
implies that SAST behaves like an oracle that sees all data points (including future ones). Our numerical results show that the FDR levels of SAST are conservative at the beginning but the FDR becomes closer to $\alpha$ as we sequentially update the barrier with information from more time points.

		\begin{theorem}
Assume conditions from Theorem \ref{thm2} holds. Then under the model with $\pi_t\equiv \pi$ and $f_t\equiv f$, the data-driven barrier $\hat\gamma_t \goto \gamma_{OR}$ when $t\goto\infty$, where $\gamma_{OR}$ is the optimal threshold of the oracle FDR procedure for simultaneous testing defined in Section \ref{oracle-rule.subsec}.   
		\end{theorem}

\section{Simulation}

In this section, we first provide some details in implementation. Simulation studies are conducted in Section 4.2 to compare the oracle and data-driven SAST procedures with other existing online FDR rules. Section 4.3 presents an example to illustrate the merit of including a barrier in online sequential testing.

\subsection{Implementation Details}\label{implement:sec}

In our simulation, the conditional density function $\hat f_t(x)$ is estimated using R function \texttt{density}, where the bandwidths $h_x$ and $h_t$ are chosen based on \cite{Sil86}. A key step in the SAST algorithm is to estimate $\hat\pi^\tau$. We propose to choose a data-driven $\tau_{BH}$ by running BH at $\alpha=0.5$. Roughly speaking, in the subset $\tilde{\mathcal T}_t(\tau_{BH})=\{t-d+1\leq i \leq t-1: P_{i}<\tau\}$, 50\% of the cases come from the null (e.g. the expected proportion of false positives made by BH). It is anticipated that in the remaining set ${\mathcal T}_t(\tau)=\{t-d+1\leq i \leq t-1: P_{i}>\tau_{BH}\}$, which is used to construct our estimator, majority of the cases should come from the null. This data-driven scheme ensures a small bias in approximation, while maintaining a larger sample size compared to the standard choice of $\tau=0.5$. 
\subsection{Comparisons of online FDRs and MDRs}

We compare the proposed SAST procedure with its competitors for online FDR control. The following methods are included in the comparison: 

\begin{itemize}

 \item SAST with known $\pi_t$ and $f_t$ (SAST.OR, Algorithm 1) 
 
 \item SAST with estimated model parameters (SAST.DD, Algorithm 2)

 \item LOND: the method proposed by Javanmard, A. and Montanari, A. (2016). 
 
 \item LORD++: the GAI++ rule proposed by \cite{Rametal17}.  
\end{itemize}

For the general simulation setup, we choose $m=5000$ and the pre-specified FDR level $\alpha=0.05$. The data are simulated from the following model:
$$
X_t \sim (1-\pi_{t})N(0,1)+\pi_{t} N(\mu, 1).
$$
For the data--driven method, we need an initial burn--in period. In simulation we generate $500$ data points prior to $t=1$ to form an initial density estimate. The varying density and proportion estimates are updated every $200$ time points. The following simulation settings are considered:

\begin{enumerate}
\item
\textbf{Block Pattern}: $\pi_t=0.01, \; \mbox{for} \;t\in (1,1000]\cup(1200,2000]\cup(2200,3000]\cup(3200,4000]\cup(4200,5000]$; $\pi_t=0.6, \; \mbox{for} \;t\in(1000,1200]\cup(2000,2200]$; $\pi_t=0.8,\; \mbox{for} \;t\in (3000,3200]\cup(4000,4200]$. Vary $\mu$ from 2 to 4.2 with step size $0.2$. 
\item
\textbf{Constant Pattern}: $\pi_t=0.05,\;t=1,\cdots,m$. Vary $\mu$ from 2 to 4.2 with step size $0.2$.
\item
\textbf{Linear Pattern}: Vary $\pi_t$ linearly from 0 to 0.5. Vary $\mu$ from 2 to 4.2 with step size $0.5$.
\item
\textbf{Sine Pattern}: $\pi_t=(\sin\frac{2\pi t}{m}+1)/4$, $\pi_t$ ranges between 0 to 0.5, vary $\mu$ from 2 to 4.2 with step size $0.5$.
\end{enumerate}

We apply different methods at $\alpha=0.05$. The empirical FDR and MDR levels are evaluated using the average of the false discovery proportions and missed discovery proportions from 1000 replications. To investigate the performance of different methods in the online setting, we display the empirical $\mbox{FDR}^t$ and $\mbox{MDR}^t$ levels at various  time points, where the intermediate evaluation points ranges from $1500$ to $5000$ with step size 500. The results for block and constant patterns are summarized in Figure \ref{block-constant.fig}, and the results for the linear and sine patterns are summarized in Figure \ref{linear-sin.fig}. 

\begin{figure}
\centering
\includegraphics[scale=0.5]{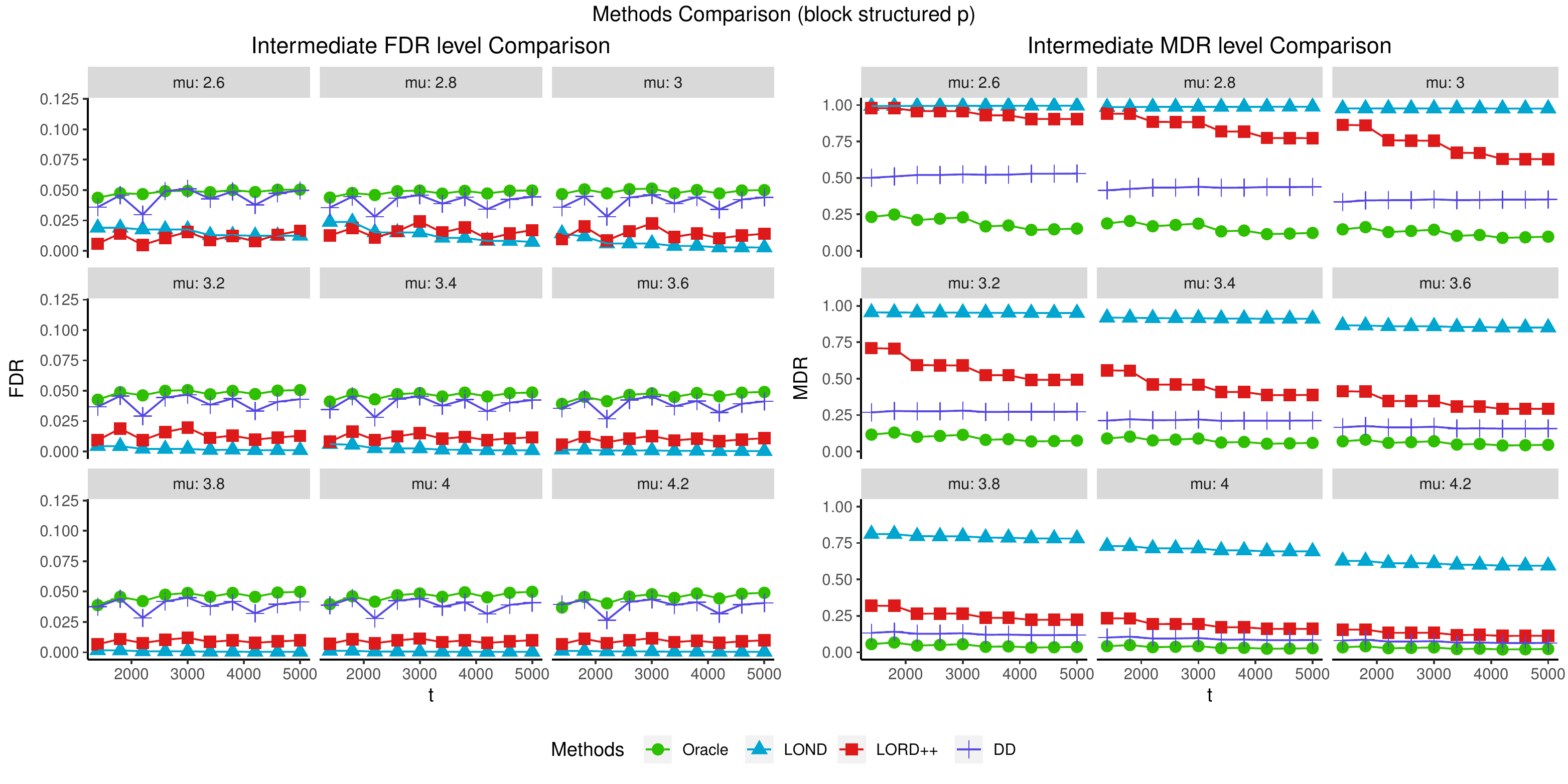}
\includegraphics[scale=0.5]{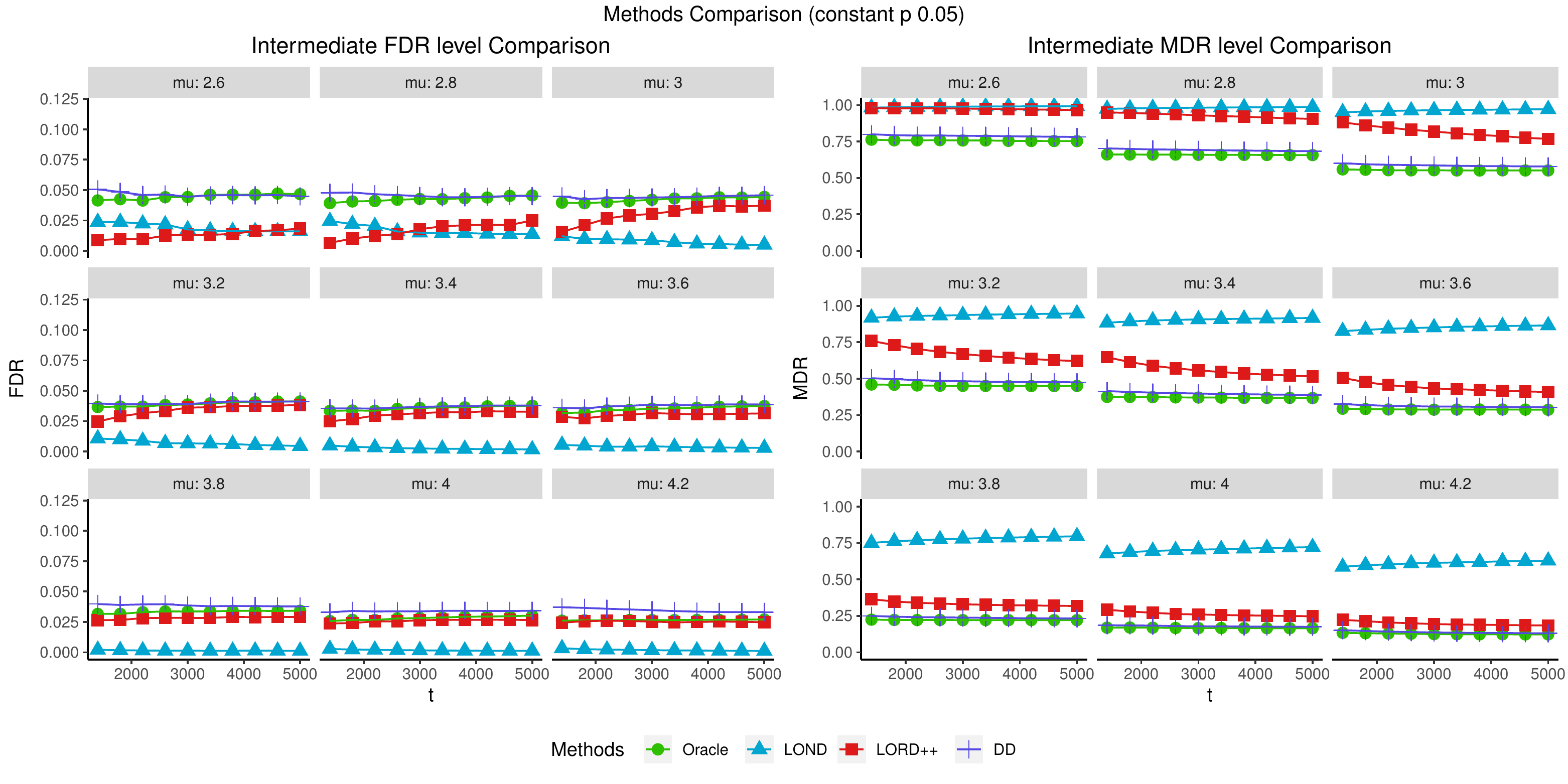}
\caption{Simulation results for Settings 1 and 2: signal proportions are varied in a block fashion and kept constant respectively. Various signal strengths are investigated as well. Our data-driven and oracle procedures provide significantly more power while controlling FDR under the nominal level in comparison with others. }\label{block-constant.fig}
\end{figure}

\begin{figure}
\centering
\includegraphics[scale=0.5]{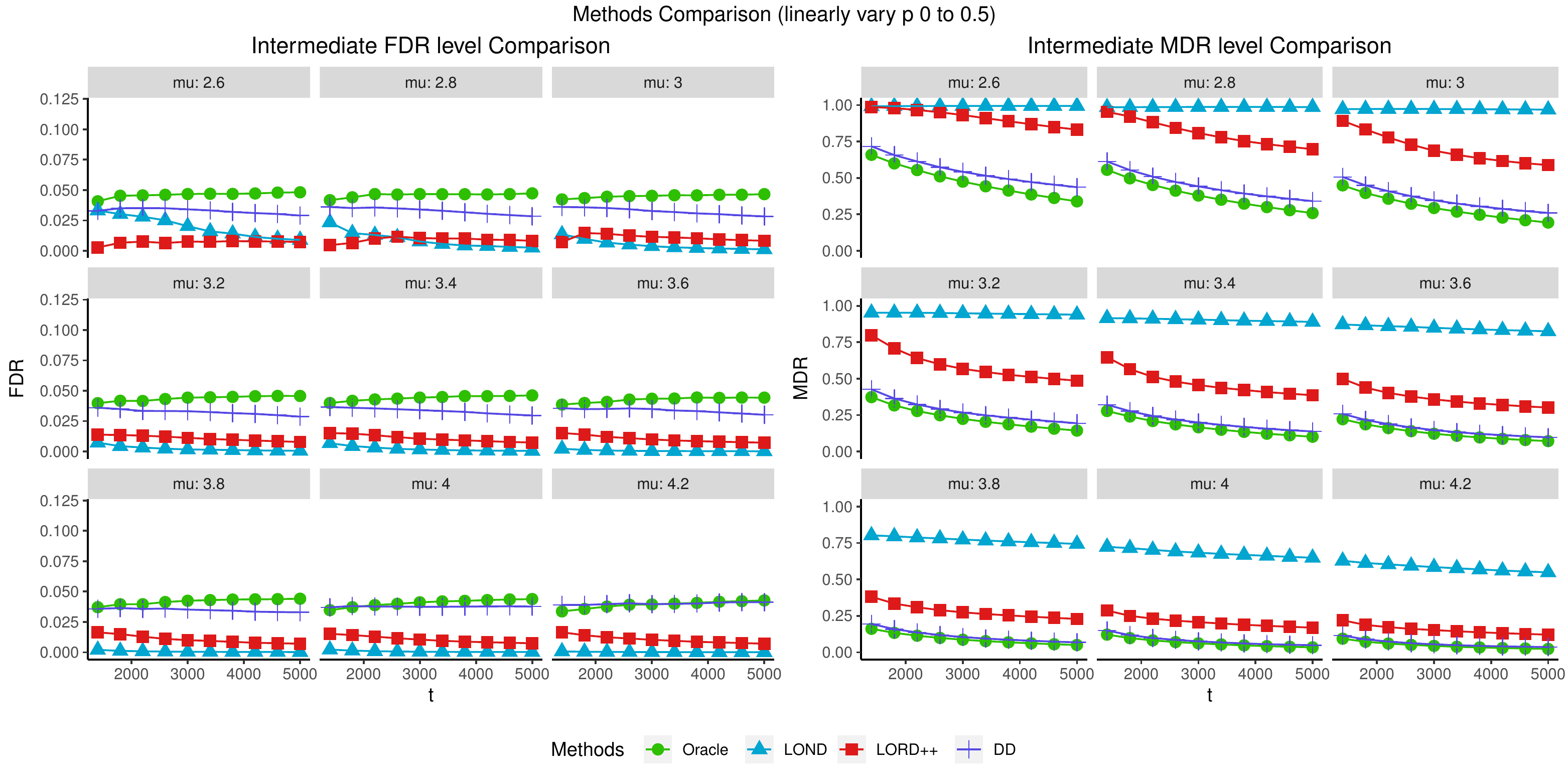}
\includegraphics[scale=0.5]{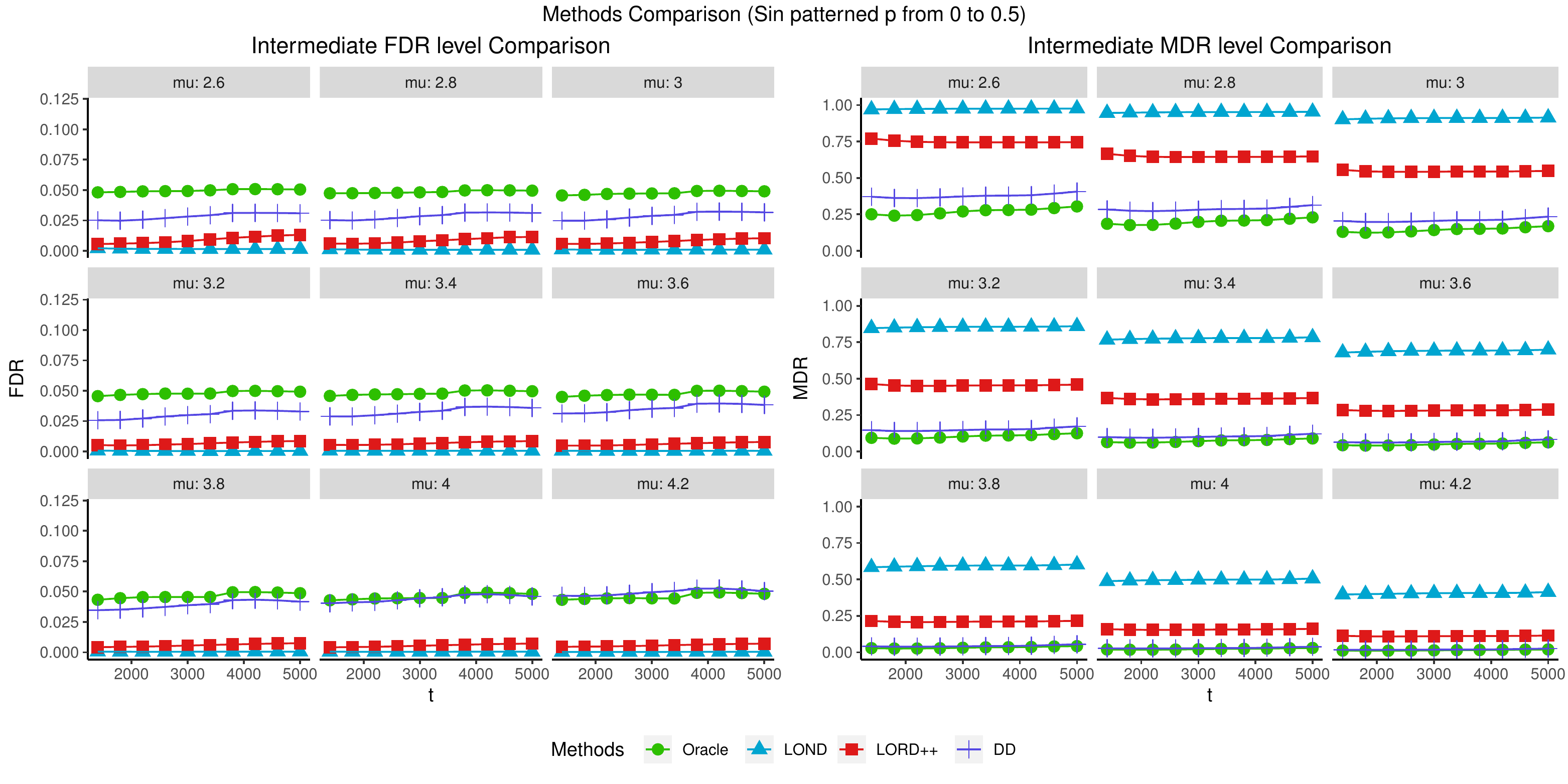}
\caption{Simulation results for Settings 3 and 4: signal proportions are varied in linear and sine patterns, respectively. Our data-driven and oracle procedures provide significantly more power while controlling FDR under the nominal level in comparison with others. }\label{linear-sin.fig}
\end{figure}

The following observations can be made from the simulation results. 
\begin{enumerate}

 \item [(a)] All methods control $\mbox{FDR}^t$ at the nominal level at all decision points being considered. SAST.OR achieves the nominal level very precisely. SAST.DD is conservative. LOND and LORD++ are more conservative compared to SAST.DD.

 \item  [(b)]  SAST.DD is inferior compared to SAST.OR. This is largely due to the conservativeness of the estimator $\hat\pi^\tau_t$. The gap in the performances between SAST.DD and SAST.OR narrows as the signal strength becomes stronger, in which situation the estimator $\hat\pi^\tau_t$ becomes more precise.

 \item  [(c)] In general LOND can be much improved by LORD++, which can be further improved by SAST.DD. The gap in power performances between SAST.DD and LORD++ narrows as the signal strength becomes stronger, in which situation it is easier to separate the signals from null cases. 
 
 \item  [(d)] When $\pi_t$ is fixed over time, the signals arrive at a constant rate and there is no informative structural information in the data stream (Setting 2: constant pattern). SAST.DD still outperforms LOND and LORD++ because our AI framework based on Clfdr precisely characterizes the gains and losses of different decisions; this not only leads to more precise FDR control but also optimizes the alpha--wealth allocation in the online setting. 
 
 
\end{enumerate}

\subsection{Effects of the barrier} \label{bar:sec}

This section presents a toy example to illustrate that the barrier, which aims to prevent the ``piggybacking'' issue (\citealp{Rametal17}), can greatly reduce the MDR by allocating existing alpha--wealth in a more cost--effective way. Consider the previous block structured setting (Setting 1 in Section 4.2). Figure \ref{fig:figbar} shows the FDR and MDR comparisons for the following methods at FDR level $\alpha=0.05$: (i) oracle SAST rule (OR); (ii) oracle SAST rule with no barrier (OR\_\,nob); (iii) data-driven SAST rule with estimated parameters (DD, Section 3); (iv) data-driven SAST rule with no barrier (DD\_\,nob).

\begin{figure}[ht]
	\includegraphics[scale=0.5]{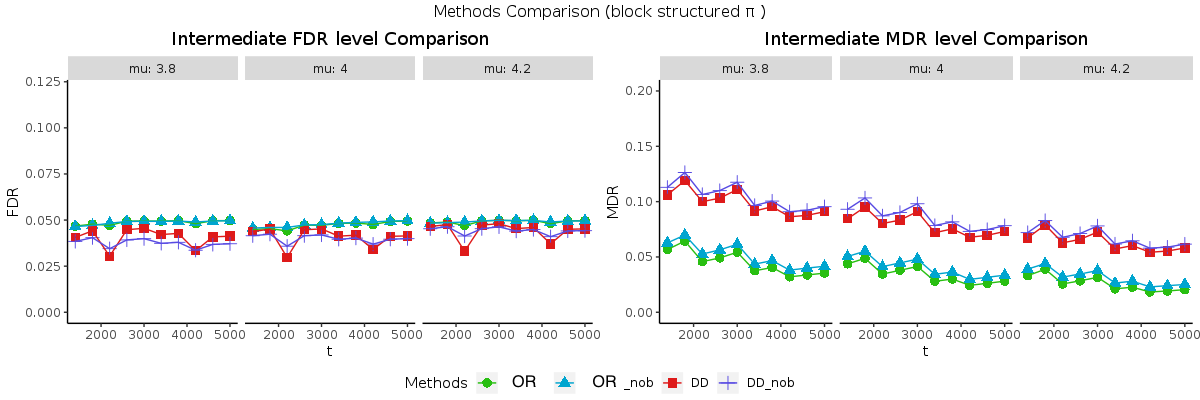}\caption{The incorporation of the barrier greatly reduces the MDR levels. }
	\label{fig:figbar}
\end{figure}

We can see from the comparison that although the FDR levels between the two oracle methods are roughly the same, the MDR levels are greatly reduced by incorporating the barrier (hence the alpha--wealth is invested  more efficiently). The same patterns can be observed for the two data-driven procedures.

\section{Applications}\label{app:sec}

Online FDR rules are useful for a wide range of scenarios. We discuss two applications, respectively for anomaly detection in large--scale time series data  and genotype discovery under the QPD framework.

\subsection{Time series anomaly detection}

The NYC taxi dataset can be downloaded from the Numenta Anomaly Benchmark (NAB) repository \citep{Ahmad17}, which contains useful tools and datasets for evaluating algorithms for anomaly detection in streaming, real--time applications. The dataset records the counts of NYC taxi passengers every 30 minutes from July 1, 2014 to January 31, 2015, during which period five known anomalies had occurred (the NYC marathon, Thanksgiving, Christmas, New Years day and a snow storm). In Figure \ref{fig:nyc_taxi}, we plot the time series, with the known anomalous intervals displayed in red rectangles. 

\begin{figure}[ht]
\centering
	\includegraphics[scale=0.37]{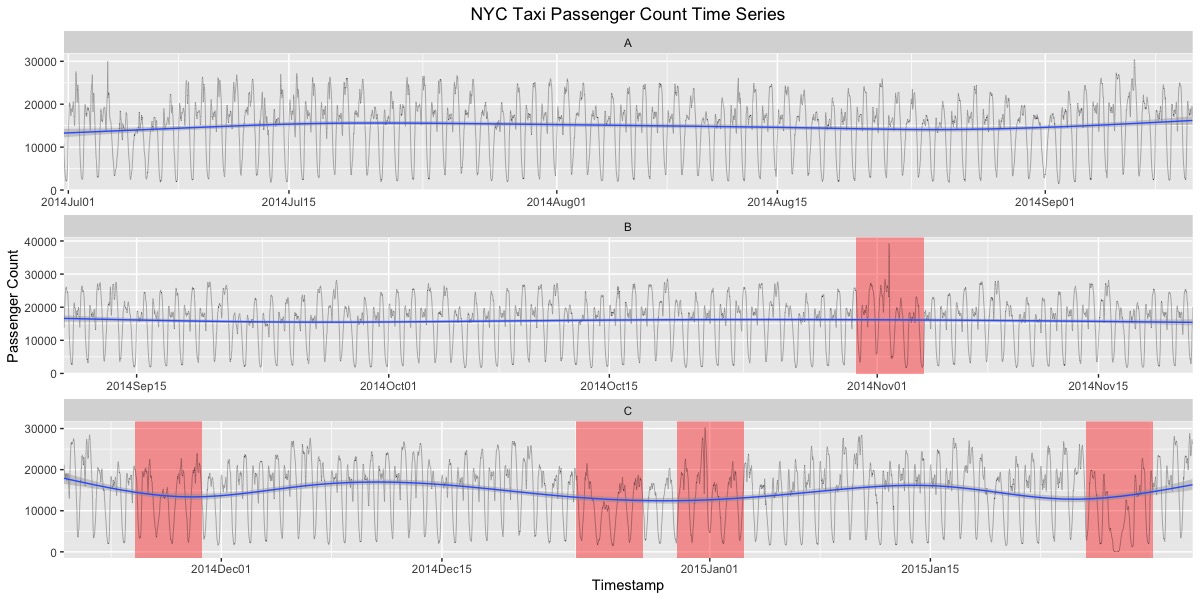}\caption{NYC Taxi passenger count time series from July 1st 2014 to Jan 31st 2015. Blue lines are Loess smoothed time series indicating the overall trend change. }
	\label{fig:nyc_taxi}
\end{figure}

We formulate the anomaly detection problem as an online sequential multiple testing problem. The basic setup can be described as follows. The null hypothesis $H_t$ corresponds to no anomaly at time $t$. We claim that an anomaly occurs at $t$ if $H_t$ is rejected. A rejection within the red intervals is considered to be a true discovery. 

The application of online FDR rules requires summarizing the stream of counts data as a sequence of $p$-values or CLfdr statistics. However, directly calculating the $p$-values based on this dataset would be problematic as the data demonstrate strong trend and seasonality patterns. We first use the R package \texttt{stlplus} to carry out an STL decomposition (Seasonal Trend decomposition using Loess smoothing; \citealp{Cle90}) to remove the seasonal and trend components. The residuals, displayed in the top 3 rows of Figure \ref{fig:remainder},
 are standardized and modeled using a two-component mixture \eqref{hmix-model}. However, as can be seen from the histogram at the bottom of Figure \ref{fig:remainder}, the null distribution is approximately normal but deviates from a standard normal. Following the method in \cite{JinCai07}, we estimate the empirical null distribution as $N(0.028,0.618)$. We apply the BH (pretending all observations are seen at once), LOND, LORD++ and SAST.DD at FDR level 0.0001. For the SAST.DD method, the neighborhood size $d$ and initial burn-in period are both chosen to be $500$. In calculating the Clfdr, $f_0(x)$ is taken as the density of the estimated empirical null $\hat F_0$. Moreover, the $p$-values are obtained by the formula $P_i=2\hat F_0(-|Z_i|)$, where $z$-scores are computed based on the residuals. Figure \ref{fig:anomaly_detected} summarizes the anomaly points detected by different methods. 

We can see that for the several anomaly time periods labeled, SAST can detect more points than other methods. Table \ref{table:anomaly_detected} summarizes the total number of rejections within the labeled time windows. It may appear counter-intuitive that SAST, being an online procedure, rejects more null hypotheses than the offline BH procedure. The reason is that the anomalies tend to appear in clusters. This structural information is captured by the Clfdr statistic, which forms the building block of SAST and leads to improved power in detecting structured signals (Section \ref{adapt:sec}).

 \begin{figure}
\centering
	\includegraphics[width=6in, height=3.5in]{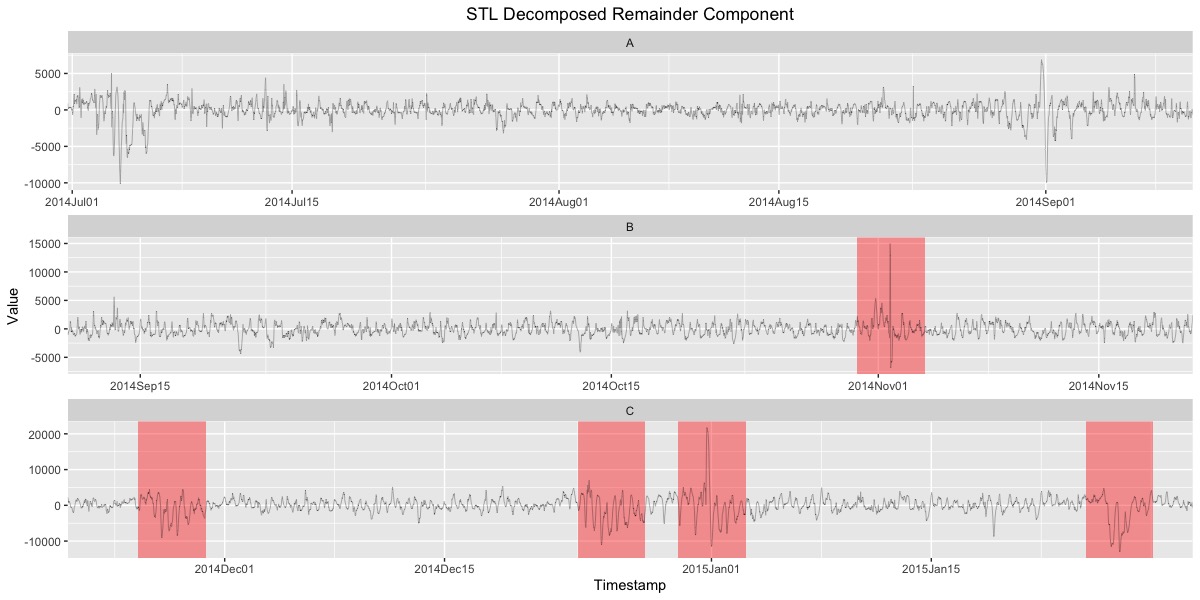}
			\includegraphics[width=6in, height=2in]{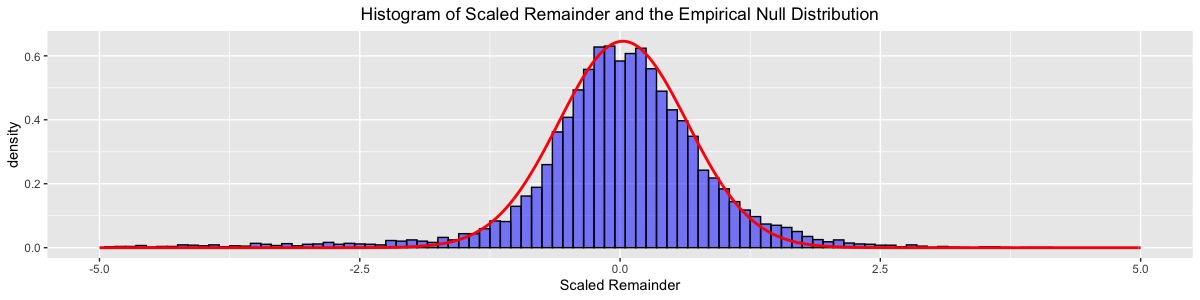}
\caption{Top three rows: Time series of remainder component from STL decomposition with the known anomaly regions marked in red rectangles. Bottom row: Histogram of the remainder term from STL decomposition, the red curve indicates the estimated empirical null distribution $N(0.028,0.618)$.}\label{fig:remainder}
\end{figure}

\begin{figure}[!ht]
	\centering
	\includegraphics[scale=0.37]{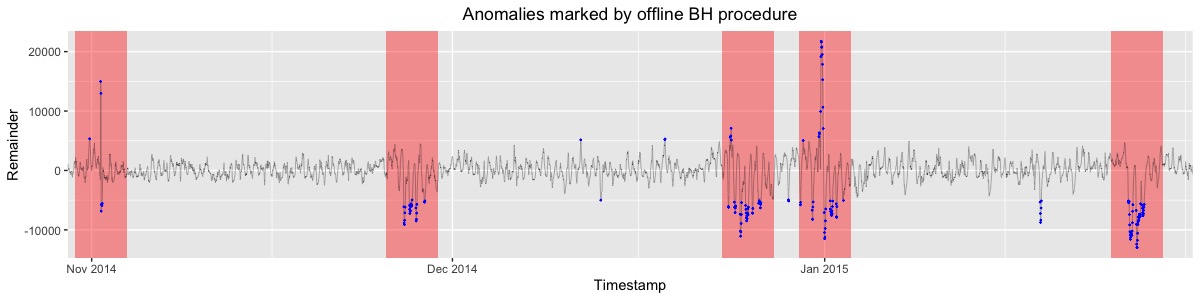}
	\includegraphics[scale=0.37]{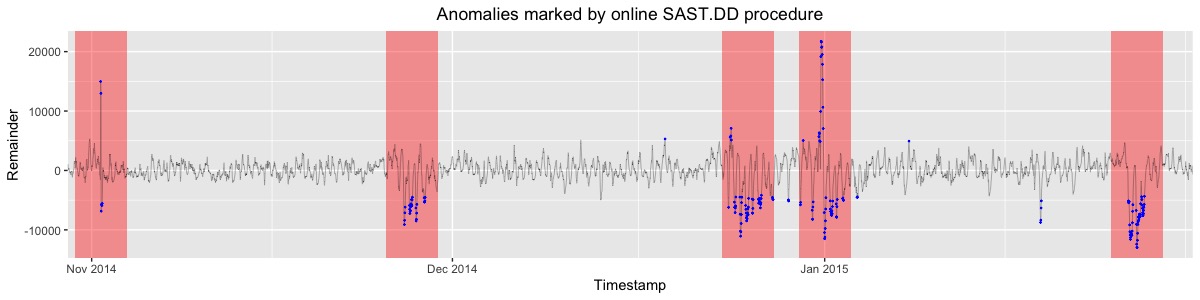}
	\includegraphics[scale=0.37]{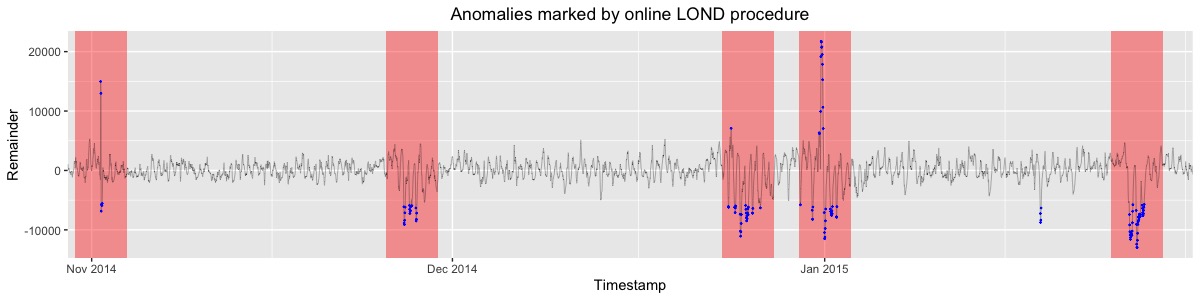}
	\includegraphics[scale=0.37]{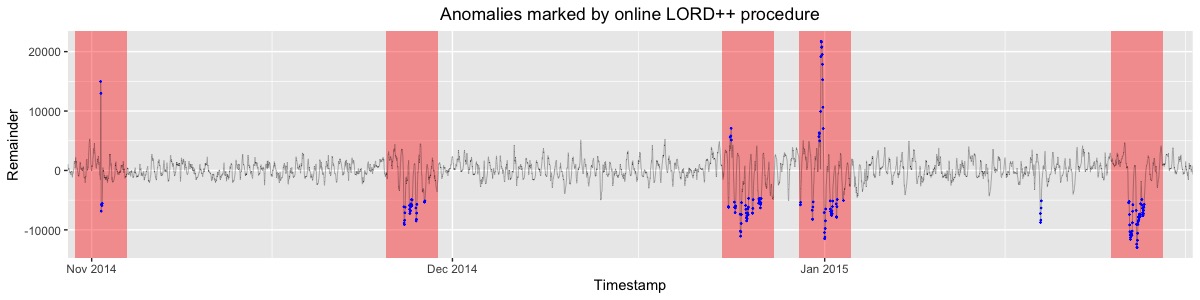}
	\caption{Anomaly points detected by various algorithms, our data-driven SAST procedure detects the most anomaly points within the labeled window marked by red rectangles. Nominal significance level chosen as 0.0001. }
	\label{fig:anomaly_detected}
\end{figure}

\begin{table}[]
\centering
\begin{tabular}{|c|c|}
\hline
\textbf{Method} & \textbf{Number of Discoveries} \\ \hline
      Offline BH procedure & 179  \\ \hline
     Online SAST.DD (Proposed) &       201  \\ \hline
	 Online LOND &   137 \\ \hline
     Online LORD++ &   178                \\ \hline
\end{tabular}
     \caption{Table 1: Number of discoveries made by various online and offline FDR procedures for the NYC taxi dataset, nominal FDR level at $0.0001$.}
     \label{table:anomaly_detected}
\end{table}

\subsection{IMPC dataset Genotype Discovery}

{In this section, we demonstrate the SAST procedure on a real dataset from the International Mouse Phenotyping Consortium (IMPC). This dataset, which has been analyzed in \cite{Kar17}, involves a large study to functionally annotate every protein coding gene by exploring the impact of gene knockouts. This dataset and resulting family of hypotheses are constantly growing as new results come in. \cite{Kar17} tested both the roles of genotype and sex as modifiers of genotype effects, resulting in two sets of $p$-values: one set for testing genotype effects, and the other for sexual dimorphism. This dataset has been widely used for comparing online FDR algorithms. Currently it is available as one of the application datasets in the R-package \texttt{OnlineFDR} that implements methods such as LORD, LOND and LORD++. In order to implement our proposed SAST procedure, we need the original $z$-scores instead of $p$-values. However, the directions of effects cannot be determined based on $p$-value alone. Hence, we transform the $p$-values into $z$-scores by introducing a Bernoulli random variable to ensure asymptotic symmetry around 0:
$z=X\Phi^{-1}(p/2)-(1-X)\Phi^{-1}(p/2),$ 
where $X\sim Ber(0.5)$\footnote{We recommend that in the future the biomedical community should report, in addition to $p$-values, the effect sizes. Thus we also know the direction and magnitude of an interesting signal. In fact, converting $z$-scores to $p$-values may lead to loss of information (cf. \citealp{SunCai07}).}. 

Table \ref{IMPC:tab} summarizes the total number of discoveries made by each method. We can see that SAST makes  more discoveries than other alpha--investing methods. Similar to the analysis in Section \ref{app:sec}, SAST rejects more hypotheses than the offline BH procedure. One possible explanation is that Clfdr is more powerful than $p$-values since it captures useful structural information in the data stream.   

\begin{table}[]
\centering
\begin{tabular}{|c|c|c|c|}
\hline
\textbf{Method} & \textbf{Genotype} & \textbf{Method type}  \\ \hline
     {SAST}  &       {12975}  & {online}    \\ \hline
     BH & 12907 & offline \\ \hline
     LORD++ & 8517 & online \\ \hline
     LOND & 2905 & online  \\ \hline
     Fixed threshold $0.0001$ & 4158 & online  \\ \hline
\end{tabular}
     \caption{Table 2: Number of discoveries made by various online and offline FDR procedures for the IMPC dataset on Genotypes, nominal FDR level at $0.05$.}\label{IMPC:tab}
\end{table}

\bibliographystyle{chicago}

\newpage

\bibliography{myrefs}

\begin{thebibliography}{}

\bibitem[\protect\astroncite{Aharoni et~al.}{2010}]{Aha10}
Aharoni, E., Neuvirth, H., and Rosset, S. (2010).
\newblock The quality preserving database: A computational framework for
  encouraging collaboration, enhancing power and controlling false discovery.
\newblock {\em IEEE/ACM transactions on computational biology and
  bioinformatics}, 8(5):1431--1437.

\bibitem[\protect\astroncite{Aharoni and Rosset}{2014}]{AhaRos14}
Aharoni, E. and Rosset, S. (2014).
\newblock Generalized $\alpha$-investing: definitions, optimality results and
  application to public databases.
\newblock {\em Journal of the Royal Statistical Society: Series B (Statistical
  Methodology)}, 76(4):771--794.

\bibitem[\protect\astroncite{Ahmad et~al.}{2017}]{Ahmad17}
Ahmad, S., Lavin, A., Purdy, S., and Agha, Z. (2017).
\newblock Unsupervised real-time anomaly detection for streaming data.
\newblock {\em Neurocomputing}, 262:134--147.

\bibitem[\protect\astroncite{Benjamini and Hochberg}{1995}]{BenHoc95}
Benjamini, Y. and Hochberg, Y. (1995).
\newblock Controlling the false discovery rate: a practical and powerful
  approach to multiple testing.
\newblock {\em J. Roy. Statist. Soc. B}, \textbf{57}:289--300.

\bibitem[\protect\astroncite{Cai and Sun}{2009}]{CaiSun09}
Cai, T.~T. and Sun, W. (2009).
\newblock Simultaneous testing of grouped hypotheses: Finding needles in
  multiple haystacks.
\newblock {\em J. Amer. Statist. Assoc.}, 104:1467--1481.

\bibitem[\protect\astroncite{Cai et~al.}{2019}]{Caietal19}
Cai, T.~T., Sun, W., and Wang, W. (2019).
\newblock {CARS}: {C}ovariate assisted ranking and screening for large-scale
  two-sample inference (with discussion).
\newblock {\em J. Roy. Statist. Soc. B}, 81(2):187--234.

\bibitem[\protect\astroncite{Cleveland et~al.}{1990}]{Cle90}
Cleveland, R.~B., Cleveland, W.~S., McRae, J.~E., and Terpenning, I. (1990).
\newblock Stl: a seasonal-trend decomposition.
\newblock {\em Journal of official statistics}, 6(1):3--73.

\bibitem[\protect\astroncite{Efron}{2004}]{Efr04}
Efron, B. (2004).
\newblock Large-scale simultaneous hypothesis testing: The choice of a null
  hypothesis.
\newblock {\em Journal of the American Statistical Association},
  99(465):96--104.

\bibitem[\protect\astroncite{Foster and Stine}{2008}]{Fos08}
Foster, D.~P. and Stine, R.~A. (2008).
\newblock $\alpha$-investing: a procedure for sequential control of expected
  false discoveries.
\newblock {\em Journal of the Royal Statistical Society: Series B (Statistical
  Methodology)}, 70(2):429--444.

\bibitem[\protect\astroncite{Genovese et~al.}{2006}]{Genetal06}
Genovese, C.~R., Roeder, K., and Wasserman, L. (2006).
\newblock False discovery control with p-value weighting.
\newblock {\em Biometrika}, 93(3):509--524.

\bibitem[\protect\astroncite{Holm}{1979}]{holm1979simple}
Holm, S. (1979).
\newblock A simple sequentially rejective multiple test procedure.
\newblock {\em Scandinavian journal of statistics}, pages 65--70.

\bibitem[\protect\astroncite{Hu et~al.}{2010}]{Huetal10}
Hu, J.~X., Zhao, H., and Zhou, H.~H. (2010).
\newblock False discovery rate control with groups.
\newblock {\em Journal of the American Statistical Association},
  105(491):1215--1227.

\bibitem[\protect\astroncite{Javanmard et~al.}{2018}]{Jav16}
Javanmard, A., Montanari, A., et~al. (2018).
\newblock Online rules for control of false discovery rate and false discovery
  exceedance.
\newblock {\em The Annals of statistics}, 46(2):526--554.

\bibitem[\protect\astroncite{Jin and Cai}{2007}]{JinCai07}
Jin, J. and Cai, T.~T. (2007).
\newblock Estimating the null and the proportional of nonnull effects in
  large-scale multiple comparisons.
\newblock {\em J. Amer. Statist. Assoc.}, \textbf{102}:495--506.

\bibitem[\protect\astroncite{Karp et~al.}{2017}]{Kar17}
Karp, N.~A., Mason, J., Beaudet, A.~L., Benjamini, Y., Bower, L., Braun, R.~E.,
  Brown, S.~D., Chesler, E.~J., Dickinson, M.~E., Flenniken, A.~M., et~al.
  (2017).
\newblock Prevalence of sexual dimorphism in mammalian phenotypic traits.
\newblock {\em Nature communications}, 8:15475.

\bibitem[\protect\astroncite{Lei and Fithian}{2018}]{LeiFit18}
Lei, L. and Fithian, W. (2018).
\newblock Adapt: an interactive procedure for multiple testing with side
  information.
\newblock {\em J. Roy. Statist. Soc. B}, 80(4):649--679.

\bibitem[\protect\astroncite{Li and Barber}{2019}]{LiBar19}
Li, A. and Barber, R.~F. (2019).
\newblock Multiple testing with the structure-adaptive benjamini--hochberg
  algorithm.
\newblock {\em Journal of the Royal Statistical Society: Series B (Statistical
  Methodology)}, 81(1):45--74.

\bibitem[\protect\astroncite{Lynch et~al.}{2017}]{Lyn17}
Lynch, G., Guo, W., Sarkar, S.~K., Finner, H., et~al. (2017).
\newblock The control of the false discovery rate in fixed sequence multiple
  testing.
\newblock {\em Electronic Journal of Statistics}, 11(2):4649--4673.

\bibitem[\protect\astroncite{Ramdas et~al.}{2017}]{Rametal17}
Ramdas, A., Yang, F., Wainwright, M.~J., and Jordan, M.~I. (2017).
\newblock Online control of the false discovery rate with decaying memory.
\newblock In {\em Advances In Neural Information Processing Systems}, pages
  5650--5659.

\bibitem[\protect\astroncite{Ramdas et~al.}{2018}]{Rametal18}
Ramdas, A., Zrnic, T., Wainwright, M., and Jordan, M. (2018).
\newblock Saffron: an adaptive algorithm for online control of the false
  discovery rate.
\newblock In {\em International Conference on Machine Learning}, pages
  4286--4294.

\bibitem[\protect\astroncite{Robertson and Wason}{2018}]{Rob18}
Robertson, D.~S. and Wason, J. (2018).
\newblock Online control of the false discovery rate in biomedical research.
\newblock {\em arXiv preprint arXiv:1809.07292}.

\bibitem[\protect\astroncite{Silverman}{1986}]{Sil86}
Silverman, B.~W. (1986).
\newblock {\em Density estimation for statistics and data analysis}, volume~26.
\newblock CRC press.

\bibitem[\protect\astroncite{Stein}{2012}]{Ste12}
Stein, M.~L. (2012).
\newblock {\em Interpolation of spatial data: some theory for kriging}.
\newblock Springer Science \& Business Media.

\bibitem[\protect\astroncite{Sun and Cai}{2007}]{SunCai07}
Sun, W. and Cai, T.~T. (2007).
\newblock Oracle and adaptive compound decision rules for false discovery rate
  control.
\newblock {\em J. Amer. Statist. Assoc.}, \textbf{102}:901--912.

\bibitem[\protect\astroncite{Wand and Jones}{1994}]{WanJon94}
Wand, M.~P. and Jones, M.~C. (1994).
\newblock {\em Kernel smoothing}.
\newblock Chapman and Hall/CRC.

\bibitem[\protect\astroncite{Xia et~al.}{2020}]{Xiaetal19}
Xia, Y., Cai, T.~T., and Sun, W. (2020+).
\newblock Gap: A general framework for information pooling in two-sample sparse
  inference.
\newblock {\em J. Am. Statist. Assoc., to appear.}

\end{thebibliography}

\setcounter{equation}{0}


\newpage

\setcounter{page}{1} 

	\appendix
\begin{center}\LARGE
Online Supplementary Material for ``Structure--Adaptive Sequential Testing for Online False Discovery Rate Control''
\end{center}

\medskip

This supplement contains the proofs of main theorems (Section A), other theoretical results (Section B), and optimality theory on simultaneous testing (Section C).

	\section{Proof of main theorems}

	\subsection{Proof of Theorem 1}
	
	Note that the Clfdr is defined as $\mbox{Clfdr}_i=\mathbb{P}(\theta_i=0|X_i)$. Then by the definition of FDR and double expectation theorem, we have:
	$$\text{FDR}_t=\mathbb{E}_{\X} \left\{  (|\mathcal{R}_{t}|\vee1)^{-1}\sum_{i\in\mathcal{R}_{t}}\text{Clfdr}_i   \right\}.$$
	By construction of the decision rule, $(|\mathcal{R}_{t}|\vee1)^{-1}\sum_{i\in\mathcal{R}_{t}}\text{Clfdr}_i\leq \alpha$ for all realization of $\pmb{X} $. It follows that $\text{FDR}_t\leq \alpha$.

	\subsection{Proof of Theorem \ref{thm2} }
	We need the following lemma:
	\begin{lemma}\label{lem4}
		Suppose $a_n\xrightarrow{p}0$ and $|a_n|$ is bounded for all $n$, then $\lim\limits_{n\rightarrow \infty}\dfrac{\sum_{i=1}^{n}a_i}{n}\xrightarrow{p}0$.
	\end{lemma}
	The proof of lemma \ref{lem4} is elementary thus omitted. By definition of our algorithm, if $\widehat{\text{Clfdr}}_t\leq \alpha$ then $\delta_t=1$. Note that, for any $\epsilon>0$,
	\begin{align*}
	\mathbb{P}\left\{\sum_{i=1}^{\infty}\mathbb{I}(\text{Clfdr}_i^\tau<\alpha-\epsilon)<\infty\right\}&=	\mathbb{P}\left\{\bigcup^\infty_{M=1}\sum_{i=1}^{\infty}\mathbb{I}(\text{Clfdr}_i^\tau<\alpha-\epsilon)<M\right\}\\
	&\leq \sum_{M=1}^{\infty}\mathbb{P}\left\{ \sum_{i=1}^{\infty}\mathbb{I}(\text{Clfdr}_i^\tau<\alpha-\epsilon)<M    \right\}.
	\end{align*}
Note that $\mbox{Clfdr}_i$ is a random variable from random mixture model \eqref{hmix-model} with a non-vanishing proportion of nonzero signals, we have
$$\mathbb{P}\left\{ \sum_{i=1}^{\infty}\mathbb{I}(\text{Clfdr}_i^\tau<\alpha-\epsilon)<M  \right\}=0$$ for every $M$. We have $	\mathbb{P}\left\{\sum_{i=1}^{\infty}\mathbb{I}(\text{Clfdr}_i^\tau<\alpha-\epsilon)<\infty\right\}=0 $. Now, $\sum_{i=1}^{\infty}\mathbb{I}(\widehat{\text{Clfdr}}_i\leq \alpha)<\infty$ would imply $|\widehat{\text{Clfdr}}_i- \text{Clfdr}_i^\tau|>\epsilon $ infinitely many times. By Proposition \ref{SAWS_prop},  $\mathbb{P}(|\widehat{\text{Clfdr}}_i- \text{Clfdr}_i^\tau|>\epsilon )\rightarrow 0$. It follows that 
		$\mathbb{P}\left\{ \sum_{i=1}^{\infty}\mathbb{I}(\widehat{\text{Clfdr}}_i\leq \alpha)<\infty  \right\} =0,$ hence $|\mathcal{R}_t|\rightarrow \infty$. By  Proposition \ref{SAWS_prop} and Lemma \ref{lem4}, we have
		$$ \dfrac{\sum_{i \in \mathcal{R}_t}\widehat{\text{Clfdr}}_i- \text{Clfdr}_i^\tau}{|\mathcal{R}_t|}\xrightarrow{p} 0.$$
	
Finally, the operation of Algorithm 2 implies that $ \dfrac{\sum_{i \in \mathcal{R}_t}\widehat{\text{Clfdr}}_i}{|\mathcal{R}_t|}\leq \alpha,$
	It follows that
	$$\mbox{FDR}(\pmb{\delta})= \mathbb{E}_{\X}\left[\dfrac{\sum_{i \in \mathcal{R}_t}\text{Clfdr}_i}{|\mathcal{R}_t|}\right]\leq \mathbb{E}_{\X}\left[\dfrac{\sum_{i\in \mathcal{R}_t}\text{Clfdr}^\tau_i}{|\mathcal{R}_t|}\right]=  \mathbb{E}_{\X}\left[\dfrac{\sum_{i \in \mathcal{R}_t}\widehat{\text{Clfdr}}_i}{|\mathcal{R}_t|}\right]+o(1)\leq \alpha+o(1).$$

	\subsection{Proof of theorem 3}

		Note when both $f_t$ and $\pi_t$ are fixed over time, the Clfdr statistic reduces to $\mbox{Lfdr}_i\coloneqq \frac{(1-\pi)f_0(x_i)}{f(x_i)}$. The optimal threshold in the offline simultaneous testing setup would be independent of time $t$ and the chosen neighborhood. The oracle offline rule coincides with the oracle procedure described in Section 3.2 of \cite{SunCai07}. 
		
	We now introduce some notations: 
	\begin{itemize}
		\item $\hat{U}^t(\gamma)=t^{-1}\sum_{i=1}^{t}(\widehat{\mbox{Clfdr}}_{(i)}-\alpha)\mathbb{I}\widehat{\{\mbox{Clfdr}}_{(i)}< \gamma\} $
		\item $U^t(\gamma)=t^{-1}\sum_{i=1}^{t}(\mbox{Clfdr}^\tau_{(i)}-\alpha)\mathbb{I}\{\mbox{Clfdr}^\tau_{(i)}< \gamma\} $.
		\item  $U^t_{\infty}(\gamma) = \mathbb{E}\{(\mbox{Clfdr}^\tau-\alpha)\mathbb{I}\{\mbox{Clfdr}^\tau<\gamma\}\}$.
		\item  $\gamma_\infty=\sup\{\gamma\in(0,1), U^t_{\infty}(\gamma)\leq 0  \} $ is the ``ideal" threshold. 
	\end{itemize} 
	
Note that $\hat{U}^t$ is discrete. To facilitate the theoretical analysis, we define, for $\widehat{\mbox{Clfdr}}_{(i)}< \gamma<\widehat{\mbox{Clfdr}}_{(i+1)}$, a continuous version of $\hat{U}^t$:
	$$\hat{U}^t_C(\gamma)=\dfrac{\gamma-\widehat{\mbox{Clfdr}}_{(i)}}{\widehat{\mbox{Clfdr}}_{(i+1)}-\widehat{\mbox{Clfdr}}_{(i)}}\hat{U}^t_i+\dfrac{\widehat{\mbox{Clfdr}}_{(i+1)}-\gamma}{\widehat{\mbox{Clfdr}}_{(i+1)}-\widehat{\mbox{Clfdr}}_{(i)}}\hat{U}^t_{i+1}, $$
	where $\hat{U}^t_i=\hat{U}^t(\widehat{\mbox{Clfdr}}_{(i)})$. It is easy to verify that $\hat{U}^t_C$ is continuous and monotone. Hence its inverse $\hat{U}^{t,-1}_C$ is well defined, continuous and monotone. \\
	Next we show the following two results in turn: (i) $\hat{U}^t(\gamma)\xrightarrow{p} U^t_\infty(\gamma)$ and (ii)$\hat{U}^{t,-1}_C(0)\xrightarrow{p} \gamma_\infty$. \\

\noindent\textbf{Proof of (i).} Note that $U^t(\gamma)\xrightarrow{p}U^t_\infty(\gamma)$ by the WLLN, so that we only need to establish that $\hat{U}^t(\gamma)\xrightarrow{p}U^t(\gamma)$. We need to following lemma:
	\begin{lemma}\label{lemma2}
		Let $V_i = (\mbox{Clfdr}^\tau_i-\alpha) \mathbb{I}(\mbox{Clfdr}^\tau_i < \gamma)$ and $\hat{V}_i = (\widehat{\mbox{Clfdr}}_i-\alpha)\mathbb{I}\{\widehat{\mbox{Clfdr}}_i < \gamma\}$. Then $\mathbb{E}\left(\hat{V}_i - V_i \right)^2 = o(1)$.
	\end{lemma}

\noindent\textbf{Proof of Lemma \ref{lemma2}}.
Using the definitions of $\hat{V}_i$ and $V_i$, we can show that
\begin{align*}
\left(\hat{V}_i - V_i\right)^2 = & \left(\widehat{\mbox{Clfdr}}_{i} - \mbox{Clfdr}^\tau_{i}\right)^2\mathbb{I} \left(\widehat{\mbox{Clfdr}}_{i} \leq \gamma, \mbox{Clfdr}^\tau_{i} \leq \gamma \right) + \left(\widehat{\mbox{Clfdr}}_{i}-\alpha\right)^2\mathbb{I} \left(\widehat{\mbox{Clfdr}}_{i} \leq \gamma,\mbox{Clfdr}^\tau_{i}> \gamma \right)\\
& + \left(\mbox{Clfdr}^\tau_{i}-\alpha\right)^2\mathbb{I} \left(\widehat{\mbox{Clfdr}}_{i}> \gamma, \mbox{Clfdr}^\tau_{i} \leq \gamma \right).
\end{align*}
Let us refer to the three sums on the right hand as $I$, $II$, and $III$ respectively. By step 2 in the proof of Theorem 2, $I = o(1)$. Then let $\varepsilon > 0$, and consider that
\begin{align*}
\mathbb{P}\left(\widehat{\mbox{Clfdr}}_{i} \leq \gamma, \mbox{Clfdr}^\tau_{i}> \gamma \right) &\leq \mathbb{P}\left(\widehat{\mbox{Clfdr}}_{i} \leq \gamma, \mbox{Clfdr}^\tau_{i}\in \left(\gamma, \gamma+ \varepsilon \right) \right)+\mathbb{P}\left(\widehat{\mbox{Clfdr}}_{i} \leq t, \mbox{Clfdr}^\tau_{i}\geq  \gamma+ \varepsilon  \right) \\
&\leq \mathbb{P}\left\{\mbox{Clfdr}^\tau_{i} \in \left(\gamma, \gamma+ \varepsilon \right)\right\} + \mathbb{P}\left(\left|\mbox{Clfdr}^\tau_{i} - \widehat{\mbox{Clfdr}}_{i}\right| > \varepsilon \right)
\end{align*}
The first term on the right hand is vanishingly small as $\varepsilon \rightarrow 0$ because $\widehat{\mbox{Clfdr}}_{i}$ is a continuous random variable. The second term converges to $0$ by Proposition \ref{SAWS_prop}. Noting that $0\leq\widehat{\mbox{Clfdr}}_{i} \leq 1$, we conclude $II = o(1)$. In a similar fashion, we can show that $III = o(1)$, thus proving the lemma.

Let $S_t=\sum_{i=1}^{t}(\hat{V}_i-V_i)$, by Lemma \ref{lemma2} and the Cauchy-Schwartz inequality, 
	$$\mathbb{E}\left\{\left(\hat{V}_i-V_i\right)\left(\hat{V}_j-V_j\right)\right\} = o(1).$$
	%
	%
	It follows that 
	\begin{align*}
	Var\left( t^{-1} S_t   \right) = &  t^{-2} Var(S_t) \leq t^{-2}\sum_{i=1}^{t} \mathbb{E}\left\{ \left( \hat{V}_i - V_i\right)^2 \right\} \\
	& +O\left(\frac{1}{t^2}\sum_{i,j:i\neq j} \mathbb{E}\left\{\left(\hat{V}_i-V_i\right)\left(\hat{V}_j-V_j\right)\right\}\right)\\
	=& o(1).
	\end{align*}
	By Proposition \ref{SAWS_prop}, $\mathbb{E}\left( t^{-1} S_t   \right)\rightarrow 0$, applying Chebyshev's inequality, we obtain 
$$t^{-1} S_t =\hat{U}^t-U^t\xrightarrow{p}0,$$  establishing (i).

\noindent\textbf{Proof of (ii).} Since $\hat{U}^t_C$ is continuous, for any $\epsilon >0$, we can find $\eta>0$ such that 
	$\left|\hat{U}_{C}^{t, -1}(0) - \hat{U}_{C}^{t, -1}\left\{\hat{U}_{C}^{t}\left(\gamma_{\infty}^{}\right)\right\}\right| < \varepsilon$ 
	if 
	$\left|\hat{U}_{C}^{t}\left(\gamma_{\infty}^{}\right) \right|< \eta$. It follows that
	\[
	\mathbb{P}\left\{ \left|\hat{U}_{C}^{\tau}\left(\gamma_{\infty}^{}\right) \right|> \eta\right\} \geq \mathbb{P}\left\{\left|\hat{U}_{C}^{t, -1}(0) - \hat{U}_{C}^{t, -1}\left\{\hat{U}_{C}^{t}\left(\gamma_{\infty}^{}\right)\right\}\right| > \varepsilon\right\}.
	\]
	Proposition \ref{SAWS_prop} and the WLLN imply that $\hat{U}_{C}^{t}(\gamma) \overset{p}\rightarrow U_{\infty}^{t}(\gamma).$ Note that $U_{\infty}^{t}\left(\gamma_{\infty}^{}\right) = 0$, then,
	\[
	\mathbb{P}\left(\left|\hat{U}_{C}^{t}\left(\gamma_{\infty}^{}\right) \right|>\eta\right) \rightarrow 0.
	\]
	Hence, we have
	\beq\label{eq3}
	\hat{U}_{C}^{t, -1}(0) \overset{p}\rightarrow \hat{U}_{C}^{t, -1} \left\{\hat{U}_{C}^{t}\left(\gamma_{\infty}^{}\right)\right\} = \gamma_{\infty}^{},
	\eeq
	completing the proof of (ii).

\section{Proof of propositions}

\subsection{Proof of Proposition \ref{prop1}}

Let $\mathcal{A}_\tau=\{x: P_0(x)>\tau\}$, where $P_0(x)$ is the p-value of $x$. Then
\begin{align*}
	(1-\tau)^{-1}{\pr(P_t>\tau)}&=	(1-\tau)^{-1}\int_{\mathcal{A}_\tau}f_0(x)(1-\pi_t)+\pi_{t}f_{1t}(x)dx\\
	&\geq (1-\tau)^{-1}\int_{\mathcal{A}_\tau}f_0(x)(1-\pi_t)dx\\
	&=(1-\pi_t).
\end{align*}
	Hence 
	$\pi^\tau_t = 1 - (1-\tau)^{-1}{\pr(P_t>\tau)}\leq 1-(1-\pi_{t})=\pi_t.$ By definition of $\mbox{Clfdr}_t$, we have $\mbox{Clfdr}_t^{\tau} \geq \mbox{Clfdr}_t $.

Let  $\pmb{\delta}^{\tau}_{OR}$ be the decision rule described in Algorithm 1 with $ \mbox{Clfdr}_t^{\tau} $ used in place of $ \mbox{Clfdr}_t  $.
	Let $\mathcal{R}$ be the index set of hypotheses rejected by $\pmb{\delta}^{\tau}_{OR}$. The FDR of $\pmb{\delta}^{\tau}_{OR}$ is
	\begin{align*}
	\mbox{FDR}(\pmb{\delta}^{\tau}_{OR})&=\mathbb{E}\left\{\dfrac{\sum_{i \leq \mathcal{R}}(1-\theta_i)}{|\mathcal{R}\vee 1|}\right\}\\
	&=\mathbb{E}_{\pmb{X}}\left[ \mathbb{E}\left\{      \dfrac{\sum_{i \leq \mathcal{R}}(1-\theta_i)}{|\mathcal{R}\vee 1|} \bigg|\X    \right\}           \right]\\
	&=\mathbb{E}_{\X}\left(  \dfrac{1}{|\mathcal{R}\vee 1|} \sum_{i \in \mathcal{R}}\mbox{Clfdr}_i     \right).
	\end{align*}
	Since $\mbox{Clfdr}_t^{\tau} \geq \mbox{Clfdr}_t $, it follows that
	$$	\mbox{FDR}(\pmb{\delta}^{\tau}_{OR})\leq  \mathbb{E}_{\X}\left(  \dfrac{1}{|\mathcal{R}\vee 1|} \sum_{i \in \mathcal{R}}\mbox{Clfdr}_i^{\tau}     \right)\leq \alpha.$$
	The last inequality is due to the definition of $\pmb{\delta}^{\tau}_{OR}$ which guarantees that $$  \dfrac{1}{|\mathcal{R}\vee 1|} \sum_{i \in \mathcal{R}}\mbox{Clfdr}_i^{\tau}  \leq \alpha.$$ 
	
	\subsection{Proof of Proposition \ref{SAWS_prop}}
	Under the in-fill model, we write
	$$	\hat{f}_t(x) = \frac{\sum_{j=t-d+1}^{t-1} K_{h_t}\left({1-j/t}\right) K_{h_x}\left({x_j-x}\right)}{ \sum_{j=t-d+1}^{t-1} K_{h_t}\left({1-j/t}\right)}. $$
We first state 3 lemmas that will be proved in turn.
	
	\begin{lemma}\label{lem2}
		Under the assumption of Proposition \ref{SAWS_prop}, $\mathbb{E}\int\left\{\hat{f}_t(x)-f_t(x )\right\}^2dx\rightarrow 0 .$
	\end{lemma}

		\begin{lemma}\label{lem1}
			Under the assumptions of Proposition \ref{SAWS_prop},
		\[
		\ep\|\hat{\pi}^{\tau}_t - \pi^\tau\|^2 = \ep\int\{\hat{\pi}^{\tau}_t(x) - \pi^{\tau}(x)\}^2dx \rightarrow 0.
		\]

	\end{lemma}

		\begin{lemma}\label{lem3}
		Let $\hat{\pi}^\tau_t$, $\hat{f}_t(x)$, and $\hat{f_0}$ be estimates such that $\mathbb{E}\|\hat{\pi}^\tau_t-\pi_t^\tau\|^2\rightarrow 0$, $\mathbb{E}\|\hat{f}_t(x)-f_t(x)\|^2\rightarrow 0$, $\mathbb{E}\|\hat{f}_0-f_0\|^2\rightarrow 0$, and then $\mathbb{E}\|\widehat{\mbox{Clfdr}}_t-\mbox{Clfdr}^\tau_t\|^2\rightarrow 0$.
	\end{lemma}
	By Lemma \ref{lem2} and Lemma \ref{lem1}, together with the fact that $f_0$ is known, it follows from Lemma \ref{lem3} that 
	$\mathbb{E}\|\widehat{\text{Clfdr}}_t-\text{Clfdr}^\tau_t\|^2\rightarrow 0.$ Since convergence in second order mean implies convergence in probability, we have $$\widehat{\text{Clfdr}}_t^\tau\xrightarrow{p} \text{Clfdr}_t^\tau.$$

	\subsection{Growing domain version of Proposition \ref{SAWS_prop}}\label{grow}
	
	In the growing domain framework, Proposition \ref{SAWS_prop} takes the following form:
	
	\begin{proposition}\label{SAWS}
		Suppose:
		
		\textbf{(A1')}: For any $\epsilon>0$, $\exists T$ such that for all integers $i,j$ on the interval $\left[\floor*{ t-\sqrt{\log(t)}h_t},t\right]$,$t>T$, we have       $\int|f_i(x)-f_j(x)|dx<\epsilon$.
		
		\textbf{(A2')}: $h_x\rightarrow 0$, $h_t\leq t$ and $h_xh_t\rightarrow \infty$.
		
		
		\textbf{(A3')}: $f_j(x)<C$ and $\int |f''_j(x)|dx<C $  for all $j$.
		
		\textbf{(A4')}: $d\geq ch_t$ for some $c>0$.
		
	We have  $\widehat{\text{Clfdr}}_t\xrightarrow{p} \text{Clfdr}_t^\tau.$
	\end{proposition}
	
	The proof follows the same line as the proof of proposition \ref{SAWS_prop}, thus omitted.

\subsection{Proof of Lemma \ref{lem2}}

We first compute $\mathbb{E}\hat{f}_t(x)-f_t(x)$. Note that 
$
\mathbb{E}\khx(X_j-x)=\int K(z)f_j(x-h_xz)dz.	
$
Using Taylor expansion, we have
$$f_j(x-h_xz)=f_j(x)-h_xzf^{'}_j(x)+\frac{1}{2}h_x^2z^2f^{''}_j(x)+o(h_x^2). $$
It follows that
$$ \mathbb{E}\khx(X_j-x)-f_t(x)=f_j(x)-f_t(x)+\frac{1}{2}h_x^2f^{''}_j(x)\int z^2K(z)dz+o(h_x^2). $$
Let
$A=\sum_{j=t-d+1}^{t-1}\left\{\kht(1-j/t)\frac{1}{2}h_x^2f^{''}_j(x)\int z^2K(z)dz+\kht(1-j/t)o(h_x^2)\right\}. $
Then 
$$\mathbb{E}\hat{f}_t(x)-f_t(x)=\dfrac{\sum_{j=t-d+1}^{t-1}\kht(1-\frac{j}{t})\{f_j(x)-f_t(x)\}+A} {\sum_{j=t-d+1}^{t-1}\kht(1-\frac{j}{t})}. $$
$$\int |\mathbb{E}\hat{f}_t(x)-f_t(x) | dx= O\left\{      \dfrac{\sum_{j=t-d+1}^{t-1}\kht(1-j/t)\int |f_j(x)-f_t(x)|dx}{\sum_{j=t-d+1}^{t-1}\kht(1-j/t)}           \right\}+o(h_x^2)\rightarrow 0.$$
To see why the last expression goes to 0, note that for any $\epsilon>0,$ by Assumption (A1), we can take $\delta$ such that for all $i>(1-\delta)t$, $\int|f_i(x)-f_t(x)|dx<\epsilon$. Hence,
\begin{align*}
&\ \ \   \dfrac{\sum_{j=t-d+1}^{t-1}\kht(1-j/t)\int |f_j(x)-f_t(x)|dx}{\sum_{j=t-d+1}^{t-1}\kht(1-j/t)}\\
&=\dfrac{\sum_{j=t-d+1}^{\floor*{(1-\delta)t}}\kht(1-j/t)\int |f_j(x)-f_t(x)|dx}{\sum_{j=t-d+1}^{t-1}\kht(1-j/t)}+  \dfrac{\sum_{j=\floor*{(1-\delta)t}+1}^{t-1}\kht(1-j/t)\int |f_j(x)-f_t(x)|dx}{\sum_{j=t-d+1}^{t-1}\kht(1-j/t)}\\
&\leq  \dfrac{\sum_{j=t-d+1}^{\floor*{(1-\delta)t}}\kht(1-j/t)}{\sum_{j=t-d+1}^{t-1}\kht(1-j/t)}+\epsilon.
\end{align*}
Note that $h_{t}\rightarrow 0$, we conclude that $\sum_{j=t-d+1}^{\floor*{(1-\delta)t}}\kht(1-j/t)=O\left\{\int_{\delta}^{1}K_{h_t}(x)dx \right\}\rightarrow 0.$
Also, since $dh_t\rightarrow \infty$ as $t\rightarrow\infty$, we have $\sum_{j=t-d+1}^{t-1}\kht(1-j/t)\geq c'h^{-1}_t$ for some $c'$.\\
Thus $ \dfrac{\sum_{j=t-d+1}^{\floor*{(1-\delta)t}}\kht(1-j/t)}{\sum_{j=t-d+1}^{t-1}\kht(1-j/t)}\rightarrow 0$, and
$$\lim\limits_{t\rightarrow \infty}\dfrac{\sum_{j=t-d+1}^{t-1}\kht(1-j/t)\int |f_j(x)-f_t(x)|dx}{\sum_{j=t-d+1}^{t-1}\kht(1-j/t)}=0.  $$ 
It follows from the boundedness of $\mathbb{E}\hat{f}_t$ and $f_t(x)$ that 
\beq\label{eqbias}
\int |\mathbb{E}\hat{f}_t(x)-f_t(x) |^2 dx\rightarrow 0.
\eeq
Next we compute $\Var\left\{\hat{f}_t(x)\right\}$:
\begin{align*}
\Var \left\{\khx(X_j-x)  \right\}&=\frac{1}{h_x}\int K(z)^2f_j(x-h_xz)dz-\{f_j(x)+o(1)\}^2\\
&=\frac{1}{h_x}\int K(z)^2(f_j(x)+o(1))dz-\{f_j(x)+o(1)\}^2\\
&=O\left\{\frac{1}{h_x}\int K(z)^2dzf_j(x)\right\}+o(h_x^{-1}).
\end{align*}
Some additional calculations give 
\begin{align*}
\Var \hat{f}_t(x)&=\dfrac{\sum_{j=t-d+1}^{t-1}\{\kht(1-j/t)\}^2O\left\{\frac{1}{h_x}\int K(z)^2dzf_j(x)\right\}} {\{\sum_{j=t-d+1}^{t-1}\kht(1-j/t)\}^2}\\
&=O\left[(th_x)^{-1}\dfrac{\int_{0}^{1}K^2_{h_t}(x)dx}{\left\{\int_{0}^{d/t}K_{h_t}(x)dx\right\}^2}            \right]\\
&=O\left[(th_xh_{t})^{-1}/  \left\{\int_{0}^{d/t}K_{h_t}(x)dx\right\}^2 \right]
\end{align*}
Therefore, by assumption (A3) and (A4), 
\beq\label{eqvar}
\int \Var\hat{f}_t(x)dx=O\left\{ (th_{t}h_x)^{-1}  \right\}\rightarrow 0.
\eeq
Since 
$\mathbb{E}\int \{\hat{f}_t(x)-f_t(x)\}^2=\int \{\mathbb{E}\hat{f}_t(x)-f_t(x)\}^2+\Var\{\hat{f}_t(x)\} dx,$
(\ref{eqbias}) and (\ref{eqvar}) together imply that $\mathbb{E}\int \{\hat{f}_t(x)-f_t(x)\}^2\rightarrow 0$.

\subsection{Proof of lemma \ref{lem1}}
Define
$\hat{\mathbb{P}}(P_t>\tau):= \frac{\sum_{i\in \mathcal T_t(\tau)}K_{h_t}\left(1-i/t\right) }{\sum_{i=t-d+1}^{t-1} K_{h_t}\left(1-i/t\right)}.$
Let $P_t$ be the p-value of $X_t$. 
We will show
\begin{equation}\label{eq1}
\ep\|\hat{\mathbb{P}}(P_t>\tau) -\mathbb{P}(P_t>\tau)\|^2  \rightarrow 0.
\end{equation}
We first rewrite the term 
$$ \frac{\sum_{i\in \mathcal T_t(\tau)} K_{h_t}\left(1-i/t\right) }{\sum_{i=t-d+1}^{t-1} K_{h_t}\left(1-i/t\right)}=\frac{\sum_{i\in \mathcal T_t(\tau)}K_{h_t}\left(1-i/t\right)\mathbb{I}(P(x_i)>\tau) }{\sum_{i=t-d+1}^{t-1} K_{h_t}\left(1-i/t\right)}.$$
By Lemma \ref{lem2} we have,
$\mathbb{E}\int_{\mathcal{B}}\hat{f}_t(x)-f_t(x)dx \rightarrow 0$
for every $\mathcal{B}$. In particular, take $\mathcal{B}=\{x:P(x)>\tau\}$ use the definition of $\hat{f}_t$
we have
\beq\label{eq5}
\frac{\sum_{j=t-d+1}^{t-1} K_{h_t}\left({1-j/t}\right)\mathbb{E}\int_{P(x)>\tau} K_{h_x}\left({x_j-x}\right)dx}{ \sum_{j=t-d+1}^{t-1} K_{h_t}\left({1-j/t}\right)} \rightarrow \mathbb{P}(P_t>\tau).
\eeq
To show the lemma, it is sufficient to show
\begin{equation}\label{eq2}
\mathbb{E}\int_{P(x)>\tau} K_{h_x}\left({x_j-x}\right)dx\rightarrow	\mathbb{E}\left\{\mathbb{I}(P(x_j)>\tau)\right\}.
\end{equation}
To see why (\ref{eq2}) implies (\ref{eq1}), note that (\ref{eq2}) implies 
\beq\label{eq4}
\frac{\sum_{i=t-d+1}^{t-1}\left({1-j/t}\right)\mathbb{E}\int_{P(x)>\tau} K_{h_x}\left({x_j-x}\right)dx}{ \sum_{i=t-d+1}^{t-1} K_{h_t}\left({1-j/t}\right)}\rightarrow \frac{\sum_{i=t-d+1}^{t-1} K_{h_t}\left({1-j/t}\right)\mathbb{E}\left\{\mathbb{I}(P(x_j)>\tau)\right\}}{ \sum_{i=t-d+1}^{t-1} K_{h_t}\left({1-j/t}\right)}. 
\eeq
Next note that
$$\mathbb{E}\left\{   \frac{\sum_{i=t-d+1}^{t-1} K_{h_t}\left(1-i/t\right)\mathbb{I}(P(x_i)>\tau) }{\sum_{i=t-d+1}^{t-1} K_{h_t}\left(1-i/t\right)}    \right\} = \frac{\sum_{i=t-d+1}^{t-1} K_{h_t}\left({1-j/t}\right)\mathbb{E}\left\{\mathbb{I}(P(x_j)>\tau)\right\}}{\sum_{i=t-d+1}^{t-1} K_{h_t}\left({1-j/t}\right)}, \; \mbox{and}$$
\begin{align*}
\Var\left\{   \frac{\sum_{i=t-d+1}^{t-1} K_{h_t}\left(1-i/t\right)\mathbb{I}(P(x_i)>\tau) }{\sum_{i=t-d+1}^{t-1} K_{h_t}\left(1-i/t\right)}    \right\} &=O\left[t^{-1}\dfrac{\int_{0}^{1}K^2_{h_t}(x)dx}{\left\{\int_{0}^{d/t}K_{h_t}(x)dx\right\}^2}            \right]\\
&=O\left[(th_{t})^{-1}/  \left\{\int_{0}^{d/t}K_{h_t}(x)dx\right\}^2 \right].
\end{align*}
By (A4), we have 
$\left\{\int_{0}^{1}K_{h_t}(x)dx\right\}^2\geq c$
for some constant $c>0$. Now
$th_{t}\rightarrow \infty$,
implies $\Var\left\{   \frac{\sum_{i=t-d+1}^{t-1} K_{h_t}\left(1-i/t\right)\mathbb{I}(P(x_i)>\tau) }{\sum_{i=t-d+1}^{t-1} K_{h_t}\left(1-i/t\right)}    \right\}\rightarrow 0$.\\
By Chebyshev's inequality, 
$$  \frac{\sum_{i=t-d+1}^{t-1} K_{h_t}\left(1-i/t\right)\mathbb{I}(P(x_i)>\tau) }{\sum_{i=t-d+1}^{t-1} K_{h_t}\left(1-i/t\right)} \xrightarrow{p}  \frac{\sum_{i=t-d+1}^{t-1} K_{h_t}\left({1-j/t}\right)\mathbb{E}\left\{\mathbb{I}(P(x_j)>\tau)\right\}}{ \sum_{i=t-d+1}^{t-1} K_{h_t}\left({1-j/t}\right)}.$$
Combining (\ref{eq4}) , (\ref{eq5}), (A1) and (A2),
$$\hat{\mathbb{P}}(P_t>\tau) =\frac{\sum_{i=t-d+1}^{t-1} K_{h_t}\left(1-i/t\right)\mathbb{I}(P(x_i)>\tau) }{\sum_{i=t-d+1}^{t-1} K_{h_t}\left(1-i/t\right)}\xrightarrow{p} \mathbb{P}(P_t>\tau).$$
Therefore (\ref{eq1}) follows.

We now show (\ref{eq2}). Let $\epsilon=\sqrt{h_x}$. Write
$$\mathbb{E}\int_{P(x)>\tau} K_{h_x}\left({x_j-x}\right)dx=\mathbb{E}\int_{P(x)>\tau,|x_j-x|<\epsilon} K_{h_x}\left({x_j-x}\right)dx+\mathbb{E}\int_{P(x)>\tau,|x_j-x|>\epsilon} K_{h_x}\left({x_j-x}\right)dx. $$
Use the normal tail bound,
\begin{align*}
\int_{P(x)>\tau,|x_j-x|>\epsilon} K_{h_x}\left({x_j-x}\right)dx&\leq \int_{|x_j-x|>\epsilon} K_{h_x}\left({x_j-x}\right)dx\\
&\leq 2\exp\{-1/(2h_x)\}\rightarrow 0.
\end{align*}
Define $\mathcal{A}_{\tau}=\{x_j:P(x_j)>\tau\}$, let $f_j$ be the density function for $X_j$. Note that
\begin{align*}
\mathbb{E}\int_{P(x)>\tau,|x_j-x|<\epsilon} K_{h_x}\left({x_j-x}\right)dx&=\int_{-\infty}^{\infty}\int_{P(x)>\tau,|x_j-x|<\epsilon} K_{h_x}\left({x_j-x}\right)f_j(x_j)dxdx_j\\
&=\int_{\mathcal{A}_\tau\pm \epsilon}\int_{|x_j-x|<\epsilon}K_{h_x}\left({x_j-x}\right)dxf_j(x_j)dx_j\\
&=\int_{\mathcal{A}_\tau\pm \epsilon}[1-O\{2\exp(-1/(2h_x))\}]f_j(x_j)dx_j\\
&\rightarrow \int_{\mathcal{A}_\tau}f_j(x_j)dx_j=\mathbb{E}\left\{\mathbb{I}(P(x_j)>\tau)\right\}.
\end{align*}
Hence (\ref{eq2}) is proved. The lemma follows.
\subsection{Proof of  lemma \ref{lem3}}

Note that $f_t(x)$ is continuous and positive on the real line, then there exists $K_1=[-M,M]$ such that $\mathbb{P}(x\in K_1^c)\rightarrow 0$ as $M\rightarrow \infty$.

Let $\inf_{x\in K_1}f_t(x)=l_0$ and $A_\epsilon^{f_t}=\{x:|\hat{f}_t(x)-f_t(x)|\geq l_0/2\}$. Note that 
$$\mathbb{E}\|\hat{f}_t(x)-f_t(x)\|^2\geq (l_0/2)^2\mathbb{P}(A_\epsilon^{f_t}),\text{ then }\mathbb{P}(A_\epsilon^{f_t})\rightarrow 0,$$
we claim that $f_t$ and $\hat{f}_t$ are bounded below by a positive number for large $t$ except for an event that has a low probability. Similar arguments can be applied to the upper bound of $\hat{f}_t$ and $f_t$, as well as the cases for $f_0$ and $\hat{f}_0$. Therefore, we conclude that $f_0,\;\hat{f}_0\;,f_t$, and $\hat{f}_t$ are all bounded in the interval $[l_a,l_b],\;0<l_a<l_b<\infty$ for large $t$ except for an event $A_\epsilon$ that has probability tends to 0. Hence
$
0<l_a<\inf_{z\in {A_\epsilon}}\min\{f_0,\hat{f}_0,f_t,\hat{f}_t\}<\sup_{z\in {A_{\epsilon}^c}}\max\{f_0,\hat{f}_0,f_t,\hat{f}_t\}<l_b<\infty.
$
Next note that 
$$
\widehat{\text{Clfdr}}_t^\tau-\text{Clfdr}_t=\frac{\hat{f}_0f_t(\pi^\tau_t-\hat{\pi}^\tau_t)+(1-\pi^\tau_t)f(\hat{f}_0-f_0)+(1-\pi^\tau_t)f_0(f_t-\hat{f}_t)}{\hat{f}_tf_t},
$$
we conclude that 
$$
\left(\widehat{\text{Clfdr}}_t^\tau-\text{Clfdr}^\tau_t\right)^2\leq c_1\left(\pi_t^\tau-\hat{\pi}^\tau_t\right)^2+c_2\left(\hat{f}_0-f_0\right)^2+c_3\left(\hat{f}_t-f_t\right)^2\text{ in }A_\epsilon^c.
$$
It is easy to see that $\|\widehat{\text{Clfdr}}_t^\tau-\text{Clfdr}_t^\tau\|^2$ is bounded by some constant $L$, then 
$$
\mathbb{E}\|\widehat{\text{Clfdr}}_t^\tau-\text{Clfdr}_t^\tau\|^2\leq L\mathbb{P}(A_\epsilon)+c_1\mathbb{E}\|\hat{\pi}^\tau_t-\pi^\tau_t\|^2+c_2\mathbb{E}\|\hat{f}_t-f_t\|^2+c_3\mathbb{E}\|\hat{f}_0-f_0\|^2.
$$
According to the assumptions, we further have that for a given $\epsilon>0$, there exists $M\in \mathbb{Z}^+$ such that we can find $A_\epsilon$, $\mathbb{P}(A_\epsilon)<\epsilon/(4L)$, and at the same time $\mathbb{E}\|\hat{\pi}^\tau_t-\pi_t^\tau\|^2<\epsilon/(4c_1)$, $\mathbb{E}\|\hat{f}_t-f_t\|^2<\epsilon/(4c_2)$, and $\mathbb{E}\|\hat{f}_0-f_0\|^2<\epsilon/(4c_3)$ for all $t\geq M$. Consequently, we have $\mathbb{E}\|\widehat{\text{Clfdr}}_t^\tau-\text{Clfdr}_t^\tau\|^2<\epsilon$ for $t\geq M$, and the desired result follows.

	\section{Optimality of the Clfdr rule in simultaneous testing}\label{opt-clfdr.sec}
	The optimality of the Clfdr rule in simultaneous testing is summarized in the following proposition. The idea in the proof essentially follows that in \cite{Caietal19}. We provide it here for completeness. 
	\begin{proposition}\label{opt.prop}
		Consider a class of decision rules $\pmb\delta(\gamma)=\{I(\mbox{CLfdr}_i<\gamma): 1\leq i\leq m\}$ for simultaneous testing of hypotheses $\{H_i: i\in \mathcal N_d(t)\}$ in the neighborhood of $t$. Denote $Q_{OR}(\gamma)$ the marginal FDR of $\pmb\delta(\gamma)$. If $\alpha<Q_{OR}(1)$, then the oracle threshold $\gamma_{OR}\coloneqq\sup\{\gamma: Q_{OR}(\gamma)\leq \alpha\}$ exists and is unique. Define the oracle rule  $\pmb\delta_{OR}=\left\{I(\mbox{CLfdr}_i\leq \gamma_{OR}): 1\leq i\leq m\right\}$. Then $\pmb\delta_{OR}$ is optimal for simultaneous testing in the sense that 
		$$
		\text{mFDR}\left(\pmb\delta_{OR}\right)\leq\alpha,\;\text{ETP}\left(\pmb{d}_*\right)\leq \text{ETP}\left(\pmb\delta_{OR}\right)\text{ for all }\pmb{d}_* \text{ such that } \mbox{mFDR}(\pmb d_*)\leq\alpha.
		$$
\end{proposition}
\begin{proof}

	The proof has two parts. In (a), we establish two properties of the testing rule that thresholds the Clfdr at an arbitrary $\gamma$, $\{\mathbb I (\mbox{Clfdr}_i < \gamma): 1 \leq i \leq m\}$. We show that it produces mFDR $< \gamma$ for all $\gamma$ and that its mFDR is monotonic in $t$. In (b) we show that when the threshold is $\gamma_{OR}$, the testing rule, $\bm{\delta}_{OR}$, exactly attains the mFDR level and is optimal amongst all valid testing procedures controls mFDR at level $\alpha$.
	
	\medskip
	
	\noindent\textbf{Part(a).}  For the testing rule $\{\mathbb I (\mbox{Clfdr}_i < \gamma): 1 \leq i \leq m\}$, let $Q_{OR}(\gamma) = \alpha_\gamma$. We first show that $\alpha_\gamma < \gamma$. Since $\mbox{Clfdr}_i = P(\theta_i = 0|X_i = x_i)$, then
	
	\beq\label{eq:EthetaTor}
	\mathbb{E}\left\{\sum_{i}^{m}(1-\theta_i)\delta_i\right\}= \mathbb{E}_{\bm{X }} \left[ \left\{\sum_{i}^{m} \mathbb{E}_{\bm{\theta|X }}(1-\theta_i) \delta_i \right\} \right] = \mathbb{E}_{\bm{X}}\left(\sum_{i}^{m}\mbox{Clfdr}_i\delta_i\right)
	\eeq
	
	\noindent where notation $E$ is the expected value taken over $(\bm{X, \theta})$,  notation $E_{\bm{X}}$ is the expectation taken over the  distribution of $(\bm{X})$, and  $\mathbb{E}_{\bm{\theta|X}}$ is the expectation taken over $\bm{\theta}$, holding $(\bm{X})$ fixed. We use \eqref{eq:EthetaTor} in the definition of $Q_{OR}(\gamma)$ to get
	
	\beq\label{eq:rewritemFdr}
	\mathbb{E}_{\bm{X}}\left\{\sum_{i=1}^{m}(\mbox{Clfdr}_i - \alpha_{\gamma})\mathbb{I}(\mbox{Clfdr}_i\leq \gamma)\right\} = 0.
	\eeq
	
	\noindent The equality above implies that $\alpha_\gamma < \gamma$. To see this, consider that all potentially non--zero terms arise when $\mbox{Clfdr}_i \leq \gamma$, and when this is the case, either (i) $\alpha \leq  \mbox{Clfdr}_i < \gamma$, (ii)  $\mbox{Clfdr}_i \leq \alpha < \gamma$,  or (iii) $\mbox{Clfdr}_i< \gamma \leq \alpha$. Notice (i) produces zero or positive terms on the LHS of \eqref{eq:rewritemFdr}, (ii) produces zero or negative terms, and (iii) produces negative terms. If $\alpha_\gamma \geq \gamma$, then only (iii) is possible, which contradicts the RHS. Thus,  the testing rule is valid.
	
	Next, we  show that $Q_{OR}(\gamma)$ is nondecreasing in  $\gamma$. That is, letting $Q(\gamma_j) = \alpha_j$, if $\gamma_1 < \gamma_2$, then $\alpha_{\gamma_1} \leq \alpha_{\gamma_2}$. We argue by contradiction.  Suppose that $\gamma_1 < \gamma_2$ but $\alpha_1 > \alpha_2$. First, it cannot be that $\mathbb I(\mbox{Clfdr}_i< \gamma_2) = 0$ for all $i$, because that implies $\alpha_1 = \alpha_2 $ (both equal $0$).  Next, since $\gamma_1 < \gamma_2$, 
	\begin{align*}
	(\mbox{Clfdr}_i - \alpha_2) \mathbb I(\mbox{Clfdr}_i < \gamma_2) = (\mbox{Clfdr}_i - \alpha_2)  \mathbb I(\mbox{Clfdr}_i < \gamma_1) + (\mbox{Clfdr}_i - \alpha_2)  \mathbb I(\gamma_1 \leq \mbox{Clfdr}_i < \gamma_2) 
	\end{align*}
	and rewrite $(\mbox{Clfdr}_i - \alpha_2) \mathbb I(\mbox{Clfdr}_i < \gamma_1)  =  (\mbox{Clfdr}_i- \alpha_1) \mathbb I(\mbox{Clfdr}_i < \gamma_1)   +  (\alpha_1 - \alpha_2)  \mathbb I(\mbox{Clfdr}_i < \gamma_1). $
	If $\alpha_2 < \alpha_1$, then 
	\begin{align} \label{eq:a.monot.1}
	(\mbox{Clfdr}_i - \alpha_2) \mathbb I(\mbox{Clfdr}_i <  \gamma_2) \geq &  (\mbox{Clfdr}_i - \alpha_1)  \mathbb I(\mbox{Clfdr}_i< \gamma_1) + (\alpha_1 - \alpha_2)  \mathbb I(\mbox{Clfdr}_i< \gamma_1) \\
	& +  (\mbox{Clfdr}_i - \alpha_1)  \mathbb I(\gamma_1 \leq \mbox{Clfdr}_i < \gamma_2). \nonumber
	\end{align}
	It follows that
	\[
	\mathbb{E} \left\{\sum_{i=1}^{m}(\mbox{Clfdr}_i - \alpha_2) \mathbb I(\mbox{Clfdr}_i < \gamma_2) \right\} > 0.
	\]
	To see this, consider the expectation of the sum over $m$ tests for the three RHS terms of \eqref{eq:a.monot.1}, which we reference as (i), (ii), and (iii) respectively. First, (i) is zero because of (\ref{eq:rewritemFdr}). Then for each $\mbox{Clfdr}_i < \gamma_2$, either (ii) is positive  because $\alpha_2 < \alpha_1$, or (iii) is positive because $\alpha_1 < \gamma_1$.
	
	However, \eqref{eq:rewritemFdr} establishes that $\mathbb{E}\left\{\sum_{i=1}^{m}(\mbox{Clfdr}_i - \alpha_2) \mathbb I(\mbox{Clfdr}_i< \gamma_2) \right\} = 0$, leading to a contradiction. Hence, $\alpha_1 < \alpha_2.$
	
	\medskip
	
	\noindent\textbf{Part(b).} The oracle threshold is defined as $\gamma_{OR} = \sup_{\gamma}\{\gamma \in (0,1): Q_{OR}(\gamma) \leq \alpha\}$. First, let $\bar{\alpha} = Q_{OR}(1)$, which represents the largest mFDR level that the oracle testing procedure can be. By part (a), $Q_{OR}(\gamma_{OR})$ is non--decreasing. Via the squeeze theorem,  for all $\alpha < \bar{\alpha}$, this implies that $Q_{OR}(\gamma_{OR}) = \alpha$.
	
	Next, consider the power of $\bm{\delta}_{OR} = \left\{\mathbb I(\mbox{Clfdr}_i < \gamma_{OR}): 1 \leq i \leq m\right\}$ compared to that of an arbitrary decision rule $\bm{d}_{*} = (d_{*}^{1}, \hdots, d_{*}^{m})$ such that $mFDR(\bm{d}_{*}) \leq  \alpha$. Using the previous result from part(a), it follows that
	\[
	\mathbb{E}\left\{\sum_{i=1}^{m}(\mbox{Clfdr}_i - \alpha) \delta_{OR}^{i} \right\} = 0 \qquad \text{and} \qquad \mathbb{E} \left\{\sum_{i=1}^{m}(\mbox{Clfdr}_i - \alpha) d_{*}^{i} \right\} \leq 0.
	\]
	Take the difference of the two expressions to obtain
	\beq\label{eq:or.power1}
	\mathbb{E} \left\{\sum_{i=1}^{m}(\delta_{OR}^{i} - d_{*}^{i})(\mbox{Clfdr}_i - \alpha)  \right\} \geq 0.
	\eeq
	Next apply a transformation $f(x) = (x-\alpha)/(1-x) $ to each $\delta_{OR}^{i}$. Note that because $f'(x) = (1-\alpha)/(1-x)^2 >0$, $f(x)$ is monotonically increasing. Then order is preserved: if  $\mbox{Clfdr}_i < \gamma_{OR}$ then $f(\mbox{Clfdr}_i) < f(\gamma_{OR})$ and likewise for $\mbox{Clfdr}_i^{i} > \gamma_{OR}$. This means we can rewrite  $\delta_{OR}^{i} = \mathbb I\left[\left\{ (\mbox{Clfdr}_i - \alpha)/(1-\mbox{Clfdr}_i)\right\} < \gamma_{OR}\right]$, where $\gamma_{OR} =(\gamma_{OR} - \alpha)/(1-\gamma_{OR})$. It will be useful to note that, from part (a), we have $\alpha < \lambda_{OR} < 1$, which implies that $\gamma_{OR} > 0$.
	
	Then, 
	\beq\label{eq:or.power2}
	\mathbb{E}\left[\sum_{i=1}^{m}(\delta_{OR}^{i}-d_{*}^{i})\left\{(\mbox{Clfdr}_i-\alpha) - \gamma_{OR}(1-\mbox{Clfdr}_i)\right\}\right] \leq 0.
	\eeq
	To see this,  consider that if  $\delta_{OR}^{i}-d_{*}^{i} \neq 0$, then either  (i) $\delta_{OR}^{i} > d_{*}^{i}$ or (ii) $\delta_{OR}^{i} < d_{*}^{i}$. If (i), then $\delta_{OR}^{i} = 1$ and it follows that $\left\{(\mbox{Clfdr}_i - \alpha)/(1-\mbox{Clfdr}_i)\right\} < \gamma_{OR}$. If (ii), then $\delta_{OR}^{i} = 0$ and  $\left\{(\mbox{Clfdr}_i - \alpha)/(1-\mbox{Clfdr}_i)\right\} \geq \gamma_{OR}$. For both cases, 
	$$(\delta_{OR}^{i}-d_{*}^{i})\left\{(\mbox{Clfdr}_i-\alpha) - \gamma_{OR}(1-\mbox{Clfdr}_i)\right\} \leq 0.$$ 
	Summing over all $m$ terms and taking the expectation yields \eqref{eq:or.power2}.
	
	Combine \eqref{eq:or.power1} and \eqref{eq:or.power2} to obtain
	\[
	0 \leq \mathbb{E} \left\{\sum_{i=1}^{m}(\delta_{OR}^{i} - d_{*}^{i})(\mbox{Clfdr}_i^{i} - \alpha)  \right\} 
	\leq \gamma_{OR} \mathbb{E} \left\{\sum_{i=1}^{m}(\delta_{OR}^{i} -d_{*}^{i})(\mbox{Clfdr}_i - \alpha)  \right\}. 
	\]
	Finally, since $\gamma_{OR} > 0$, it follows that $\mathbb{E} \left\{\sum_{i=1}^{m}(\delta_{OR}^{i} - d_{*}^{i})(\mbox{Clfdr}_i - \alpha)  \right\} > 0$. After distributing the $(\delta_{OR}^{i} - d_{*}^{i})$ term and separating the expectations for the sums of the two  decision rules, we apply the definition of $ETP(\bm{\delta}) =  \mathbb{E} \left\{\sum_{i=1}^{m}\delta^{i}\left(\mbox{Clfdr}_i - \alpha\right)  \right\}$ to conclude that $ETP(\bm{\delta}_{OR}) \geq ETP(\bm{d}_{*})$.
\end{proof}


\end{document}